%% file: ms.tex
\documentclass[12pt]{article}

\pdfoutput=1

\usepackage{fullpage}
\usepackage{amsfonts}
\usepackage{amsmath}
\usepackage{amsthm}
\usepackage[cmyk]{xcolor}
\usepackage[final]{graphicx}
\usepackage{tikz}

\newcommand{\Size}{1cm}
\tikzset{Square/.style={
    inner sep=0pt,
    text width=\Size, 
    minimum size=\Size,
    draw=black,
    align=center,
    }
}

\newtheorem{lemma}{Lemma}
\newtheorem{theorem}{Theorem}
\newtheorem{corollary}{Corollary}

\theoremstyle{remark}

\theoremstyle{remark}

\newcommand{\sat}{\mbox{\sc{SAT}}}
\newcommand{\sg}{\mathrm{sg}}

\begin{document}

\title{{Generalized Satisfiability Problems \\
via Operator Assignments}}

\author{Albert Atserias$^1$  \;\;\; Phokion G. Kolaitis$^2$  \;\;\; Simone Severini$^3$ \\ \; \\ $^1$ Universitat Polit\`ecnica de Catalunya \\
$^2$ University of California Santa Cruz and IBM Research--Almaden \\
$^3$ University College London and Shanghai Jiao Tong University}

\maketitle

\begin{abstract}
\noindent\input{abstract-long.tex}

\end{abstract}

\input{section-1-introduction-long.tex}

\input{section-2-preliminaries-long.tex}

\input{section-3-spectral-theorem-long.tex}

\input{section-4-reductions-long.tex}

\input{section-5-gaps-long.tex}

\input{section-6-further-applications-long.tex}

\input{section-7-closure-operations-long.tex}

\bigskip
\bigskip
\input{acknowledgments-long.tex}

\bibliographystyle{plain}
\bibliography{qsat-long}

\end{document}

%% file: abstract-long.tex
Schaefer introduced a framework for generalized satisfiability problems on the Boolean domain and characterized the computational complexity of such problems.
We investigate an algebraization of Schaefer's framework in which the Fourier transform is used to represent constraints by
multilinear polynomials in a unique way. The polynomial representation of constraints gives rise to a relaxation of the notion of satisfiability in which the values to variables are linear operators on some Hilbert space. For the case of constraints given by a system of linear equations over the two-element field, this relaxation has received considerable attention in the foundations of quantum mechanics, where such constructions as the Mermin-Peres magic square show that there are systems that have no solutions in the Boolean domain, but have solutions via operator assignments on some finite-dimensional Hilbert space.
We obtain a complete characterization of the classes of Boolean relations
for which there is a gap between satisfiability in the Boolean domain and the relaxation
of satisfiability via operator assignments.
To establish our main result, we adapt the notion of primitive-positive definability  (pp-definability) to our setting, a notion that has been used extensively in the study of constraint satisfaction problems. Here, we show
that  pp-definability gives rise to  gadget
reductions that preserve satisfiability gaps. We also present several additional applications of this method. In particular and perhaps surprisingly,
we show that the relaxed notion of
pp-definability in which the quantified variables are allowed to range
over operator assignments gives no additional expressive power
in defining Boolean relations.

%% file: section-1-introduction-long.tex
\section{Introduction and Summary of Results} \label{sec:introduction}

In 1978,   Schaefer \cite{Schaefer78} classified the computational complexity of generalized satisfiability problems. Each class $A$ of Boolean relations gives rise to the generalized satisfiability problem $\sat(A)$. An instance of $\sat(A)$ is a conjunction of relations from $A$ such that each conjunct has a tuple of variables as arguments; the question is  whether or not  there is an assignment of Boolean values to the variables, so that, for each conjunct, the resulting tuple of Boolean values belongs to the underlying relation.
 Schaefer's main result is a dichotomy theorem for the computational complexity of $\sat(A)$, namely, depending on $A$, either $\sat(A)$ is NP-complete or $\sat(A)$ is solvable in polynomial time.
%   Specifically, Schaefer showed that $\sat(\Gamma)$ is in PTIME if $\Gamma$ satisfies one of the following six conditions, while it is NP-complete in all other cases: (i) %every relation in $\Gamma$ is $0$-valid (it contains the tuple consisting entirely of $0$'s); (ii) every relation in $\Gamma$ is $1$-valid (it contains the tuple consisting %entirely of $1$'s); (iii) every relation in $\Gamma$ is bijunctive (it is the set of satisfying assignments of a 2CNF-formula); (iv) every relation in $\Gamma$ is Horn (it is %the set of satisfying assignments of a Horn formula); (v) every relation in $\Gamma$ is dual Horn (it is the set of satisfying assignments of a dual Horn formula); (vi) every %relation in $\Gamma$ is affine (it is the set of solutions of a system of linear equations over the two-element field).
   Schaefer's dichotomy theorem
   provided a unifying explanation for the NP-completeness of many well-known variants of Boolean satisfiability, such  as  POSITIVE 1-IN-3 SAT and  MONOTONE 3SAT; moreover, it became the catalyst for numerous subsequent investigations, including the pursuit of a dichotomy theorem for constraints satisfaction problems, a pursuit that became known as the Feder-Vardi Conjecture \cite{FV98}.

Every Boolean relation can be identified with its characteristic function, which, via the Fourier transform, can be represented  as a multilinear polynomial (i.e., a polynomial in which each variable has degree at most one) in a unique way. Moreover, in carrying out this transformation, the truth values  \emph{false} and \emph{true} are typically represented by $+1$ and $-1$,  instead of $0$ and $1$.  For example, it is easy to see that the multilinear polynomial representing the conjunction $x \land y$ of two variables $x$ and $y$ is   $\frac{1}{2}(1+x+y-xy)$.
%$\frac{x+y-xy+1}{2}$.
The multilinear polynomial representation of Boolean relations makes it possible to consider relaxations of satisfiability in which the variables take values in some suitable space, instead of the two-element Boolean algebra.  Such relaxations have been considered in the foundations of physics
%quantum information theory
several decades ago, where they  have played a role in %the study of entanglement.
singling out the differences between classical theory and quantum theory.
In particular, it has been shown  that there is a system of linear equations over the two-element field that has no solutions over $\{+1,-1 \}$, but the system of the associated multilinear polynomials has a solution in which the variables
are assigned  linear operators on a Hilbert space of dimension four. The Mermin-Peres magic square \cite{Mermin1990,Mermin1993,Peres1990} is  the most well known example of such a system. These constructions give \emph{small proofs} of the celebrated Kochen-Specker Theorem [8] on the impossibility to explain quantum mechanics via  hidden-variables \cite{Bell1966}.
%These constructions have been used to give simplified proofs of the celebrated Kochen-Specker Theorem \cite{KochenS67} and also to obtain related results concerning %contextuality and the impossibility to explain quantum theory via a classical hidden-variables theory \cite{Bell1966}.
More recently, systems of linear equations  with this relaxed notion of solvability  have been studied under the name of \emph{binary constraint systems}, and tight connections have been established between solvability and the existence of perfect strategies in non-local games that %model
make use of entanglement \cite{CleveLS16,CleveM14}.

A Boolean relation is  \emph{affine} if it is the set of solutions of a system of linear equations over the two-element field.
The collection LIN of all affine relations is prominent in Schaefer's dichotomy theorem, as it is  one of the main classes $A$ of Boolean relations for which $\sat(A)$ is solvable in polynomial time.
%Let LIN be the class of all affine Boolean relations, which, as seen earlier, is prominent in Schaefer's dichotomy theorem since $\sat(\mbox{LIN})$ is in PTIME.
The discussion in the preceding paragraph shows that  $\sat(\mbox{LIN})$ has instances that are unsatisfiable in the Boolean domain, but are satisfiable when linear operators on a Hilbert space are assigned to variables (for simplicity, from now on we will use the term ``operator assignments'' for such assignments).
Which other classes of Boolean relations exhibit such a gap between satisfiability in the Boolean domain and the relaxation of satisfiability via operator assignments? As a matter of fact, this question bifurcates into two separate questions, depending on whether the relaxation allows linear operators on  Hilbert spaces of arbitrary (finite or infinite) dimension or only on Hilbert spaces of finite dimension. In a recent breakthrough paper, Slofstra \cite{Slofstra2016} showed that these two questions are different for LIN by establishing the existence of systems of linear equations that are satisfiable by operator assignments on some infinite-dimensional Hilbert space, but are not satisfiable by  operator assignments on any finite-dimensional Hilbert space.
In a related vein, % but in the opposite direction,
 Ji \cite{Ji2013} showed that a 2CNF-formula is satisfiable in the Boolean domain if and only if it is satisfiable by an operator assignment in some finite-dimensional Hilbert space. Moreover, Ji showed that the same holds true for Horn formulas.  %and for dual Horn formulas.
Note that 2SAT, HORN SAT, and  DUAL~HORN~SAT also feature prominently in Schaefer's dichotomy theorem as, together with $\sat(\mbox{LIN})$,
which from now on we will denote by LIN SAT,
they constitute the main tractable cases of generalized satisfiability problems (the other tractable cases are the trivial cases of $\sat(A)$, where $A$ is  a class of $0$-valid relations or a class of  $1$-valid relations, i.e., Boolean relations that contain the tuple consisting entirely of $0$'s or, respectively, the tuple consisting entirely of $1$'s).

In this paper, we completely characterize the classes $A$ of Boolean relations for which $\sat(A)$ exhibits a gap between satisfiability in the Boolean domain and satisfiability via operator assignments.
Clearly, if every relation in $A$ is $0$-valid or every relation in $A$ is $1$-valid, then there is no gap, as every constraint is satisfied by assigning to every variable the % constant zero operator or the identity operator, respectively.
identity operator or its negation, respectively.
Beyond this, we first generalize and extend Ji's results \cite{Ji2013} by showing that if $\Gamma$ is a class of Boolean relations such that every relation in $A$ is bijunctive\footnote{A Boolean relation is \emph{bijunctive} if it is the set of satisfying assignments of a 2CNF-formula.}, or every relation in $A$ is Horn, or every relation in $A$ is dual Horn\footnote{A Boolean relation is \emph{Horn} (\emph{dual Horn}) if it is the set of satisfying assignments of a Horn (dual Horn) formula.}, then there is no gap whatsoever; this   means that an instance of $\sat(A)$ is satisfiable in the Boolean domain if and only if it is satisfiable by an operator assignment on some finite-dimensional Hilbert space if and only if is satisfiable by an operator assignment on some arbitrary  Hilbert space. In contrast, we show that for all other classes $A$ of Boolean relations, $\sat(A)$ exhibits a two-level gap: there are instances of $\sat(A)$ that are not satisfiable in the Boolean domain, but are satisfiable by an operator assignment on some finite-dimensional Hilbert space; moreover,  there are instances of $\sat(A)$ that are not satisfiable  by an operator assignment on any finite-dimensional Hilbert space, but are satisfiable by an operator assignment on some (infinite-dimensional) Hilbert space.

The proof of this result uses several different ingredients. First, we use the substitution method \cite{CleveM14} to show that there is no satisfiability gap for classes of relations that are bijunctive, Horn, and dual Horn. This gives a different proof of Ji's results \cite{Ji2013}, which were for finite-dimensional Hilbert spaces, but also shows that, for such classes of relations, there is no difference between satisfiability by linear operators on finite-dimensional Hilbert spaces and satisfiability by linear operators on arbitrary Hilbert spaces.
The main tool for proving the existence of a two-level gap for the remaining classes of Boolean relations
 is the notion of \emph{pp-definability}, that is, definability via primitive-positive  formulas, which are existential first-order formulas having a conjunction of (positive) atoms as their quantifier-free part. In the past, primitive-positive formulas have been used to design polynomial-time reductions between decision problems; in fact, this is one of the main techniques in the proof of Schaefer's dichotomy theorem.
 Here, we show that primitive-positive formulas can also be used to design \emph{gap-preserving} reductions, that is, reductions that preserve the gap between satisfiability on the Boolean domain and satisfiability by operator assignments.  To prove the existence of a two-level gap for  classes  of Boolean relations  we combine gap-preserving reductions with the two-level gap for LIN discussed earlier (i.e., the results of Mermin \cite{Mermin1990,Mermin1993}, Peres \cite{Mermin1990}, and Slofstra \cite{Slofstra2016}) and with results about Post's lattice of clones on the Boolean domain~\cite{Post1941}.

We also give two additional applications of pp-definability.
First,  we consider an extension of pp-definability in which the existential quantifiers may range over linear operators on some finite-dimensional Hilbert space. At first sight, it appears that new Boolean relations may be  pp-definable in the extended sense from a given set of Boolean relations. We show, however, that this is not the case. Specifically,
by analyzing closure operations on sets of linear operators, we show that
if a Boolean relation is pp-definable in the extended sense from other Boolean relations, then it is also pp-definable from the same relations.
In other words, for Boolean relations, this extension of pp-definability is not more powerful than standard pp-definability.
Second, we apply pp-definability to the problem of quantum realizability of contextuality scenarios.  Recently,  Fritz \cite{Fritz2016} used Slofstra's results \cite{Slofstra2016} to   resolve two problems raised by Acin et al.\ in \cite{Acin2015}. Using pp-definability and Slofstra's results, we obtain new proofs of Fritz's results that have the additional feature that the parameters involved are optimal. 

%% file: section-2-preliminaries-long.tex
\section{Definitions and Technical Background} \label{sec:definitions}

\subsection{Notation}

For an integer $n$, we write $[n]$ for the set $\{1,\ldots,n\}$.  We
use mainly the $+1,-1$ representation of the Boolean domain ($+1$ for
``false'' and $-1$ for ``true''). We write $\{ \pm 1 \}$ for the set
$\{+1,-1\}$. If $a$ denotes a tuple of length $r$ we write
$a_1,\ldots,a_r$ to denote its $r$ components. If $a$ is such a tuple
and $f$ is a function that has $a_1,\ldots,a_r$ in its domain, we
write $f(a)$ to denote the tuple $(f(a_1),\ldots,f(a_r))$. 
% A relation
% $R \subseteq D^r$ is identified with its \emph{indicator function},
% i.e.\ we think of it as a function $R: D^r \rightarrow \{ \pm 1 \}$
% defined on $(a_1,\ldots,a_r) \in D^r$ by $R(a_1,\ldots,a_r) = -1$ if
% $(a_1,\ldots,a_r)$ belongs to the relation, and $R(a_1,\ldots,a_r) =
% +1$ otherwise. The indicator function of $R$ is also known as the
% \emph{characteristic function} of $R$. Note that if $r = 1$, then $R$
% is a subset of $D$, and $R(i)$ is the indicator Boolean value telling
% whether $i$ is in $R$ or not. 
We write $\mathrm{T}$ and $\mathrm{F}$
for the full and empty Boolean relations, respectively. The letters
stand for \emph{true} and \emph{false}. Their arity is unspecified by
the notation and will be made clear by the context.

\subsection{Linear Operators and Polynomials Thereof}

Let $V$ be a complex vector space. A linear operator on $V$ is a
linear map from $V$ to $V$. The linear operator that is the identity
on $V$ is denoted by $I$, and the linear operator that is identically
$0$ is denoted by $0$. The pointwise addition of two linear operators
$A$ and $B$ is denoted by $A+B$, the composition of two linear
operators $A$ and $B$ is denoted by $AB$, and the pointwise scaling of
a linear operator $A$ by a scalar $c \in \mathbb{C}$ is denoted by
$cA$. All these are linear operators. As a result, if
$\mathbb{C}\langle X_1,\ldots,X_n \rangle$ denotes the ring of
polynomials with complex coefficients and non-commuting variables in
$X_1,\ldots,X_n$, then for a polynomial $P(X_1,\ldots,X_n)$ in
$\mathbb{C}\langle X_1,\ldots,X_n \rangle$ and linear operators
$A_1,\ldots,A_n$ on $V$, the notation $P(A_1,\ldots,A_n)$ is
explained. If $A_1,\ldots,A_n$ pairwise commute, i.e., $A_iA_j =
A_jA_i$ for all $i,j \in \{1,\ldots,n\}$, then the notation is
explained even for a polynomial in $\mathbb{C}[X_1,\ldots,X_n]$, the
ring of polynomials with commuting variables in $X_1,\ldots,X_n$.

Let $V$ and $W$ be complex vector spaces. Let $A$ be a linear operator
on $V$ and let $B$ be a linear operator on $W$. We say that $A$ and
$B$ are similar if there exists an invertible linear map $C : V
\rightarrow W$ such that $A = C B C^{-1}$. Let $A_1,\ldots,A_n$ and
$B_1,\ldots,B_n$ be linear operators on $V$ and $W$, respectively. We
say that $A_1,\ldots,A_n$ and $B_1,\ldots,B_n$ are simultaneously
similar if there exists an invertible linear map $C : V \rightarrow W$
such that $A_i = C B_i C^{-1}$ holds for all $i \in [n]$. The
following simple fact with an equally simple proof will be used
multiple times.

\begin{lemma} \label{lem:simultaneouslysimilar}
  Let $V$ and $W$ be complex vector spaces, and let
  $P(X_1,\ldots,X_n)$ be a polynomial in $\mathbb{C}\langle
  X_1,\ldots,X_n \rangle$.  If $A_1,\ldots,A_n$ and $B_1,\ldots,B_n$
  are simultaneously similar linear operators on $V$ and $W$,
  respectively, then so are $P(A_1,\ldots,A_n)$ and $P(B_1,\ldots,B_n)$.
\end{lemma}

\begin{proof}
  We write $[n]^*$ for the set of finite sequences with components in
  $[n]$, and $|\alpha|$ for the length of the sequence $\alpha$.  Let
  $P(X_1,\ldots,X_n) = \sum_{\alpha \in [n]^*} c_{\alpha}
  \prod_{i=1}^{|\alpha|} X_{\alpha_i}$, where only finitely many of the
  coefficients $c_\alpha$ are non-zero.  Let $C : V \rightarrow W$ be
  an invertible linear map witnessing that $A_1,\ldots,A_n$ and
  $B_1,\ldots,B_n$ are simultaneously similar; thus $A_j = C B_j
  C^{-1}$ holds for every $j \in [n]$. Note that for every $\alpha \in
  [n]^*$ of length $\ell$ we have $\prod_{i=1}^\ell (CB_{\alpha_i}C^{-1}) =
  C\big({\prod_{i=1}^\ell B_{\alpha_i}}\big)C^{-1}$.  It follows that
  $P(A_1,\ldots,A_n) = \sum_{\alpha \in [n]^*} c_\alpha
  C\big({\prod_{i=1}^{|\alpha|} B_{\alpha_i}}\big)C^{-1}$, and linearity gives
  $P(A_1,\ldots,A_n) = C P(B_1,\ldots,B_n) C^{-1}$.
%
%  Write $P$ as a sum of monomials: $P(X_1,\ldots,X_n) = \sum_{\alpha
%    \in \mathbb{N}^{[n]} } c_\alpha \prod_{i \in [n]}
%  X_i^{\alpha(i)}$.  Let $C : V \rightarrow W$ be an invertible linear
%  map witnessing that $A_1,\ldots,A_n$ and $B_1,\ldots,B_n$ are
%  simultaneously similar; thus $A_i = C B_i C^{-1}$ holds for every $i
%  \in [n]$. First note that for every non-negative integer $\alpha$ we
%  have $(CB_iC^{-1})^\alpha = CB_i^\alpha C^{-1}$. It follows that
%  $P(A_1,\ldots,A_n) = \sum_{\alpha \in \mathbb{N}^{[n]}} c_\alpha C
%  \big({\prod_{i \in [n]} B_i ^{\alpha(i)}}\big) C^{-1}$, and by linearlity
%  $P(A_1,\ldots,A_n) = C P(B_1,\ldots,B_n)
%  C^{-1}$.
%% Hence
%% \begin{equation}
%% P(A_1,\ldots,A_n) = \sum_{\alpha \in \mathbb{N}^{[n]}} c_\alpha
%% \prod_{i \in [n]} (C B_i C^{-1})^{\alpha(i)} =
%% \sum_{\alpha \in \mathbb{N}^{[n]}} c_\alpha
%% C \Biggl({\prod_{i \in [n]} B_i ^{\alpha(i)}}\Biggr) C^{-1}.
%% \label{eqn:simul}
%% \end{equation}
%% By linearity, the last expression in~\eqref{eqn:simul} can
%% be rewriten as
%% \begin{equation}
%% C\Biggl({\sum_{\alpha \in \mathbb{N}^{[n]}} c_{\alpha} \prod_{i \in [n]}
%%   B_i}\Biggr) C^{-1} = C P(B_1,\ldots,B_n) C^{-1}.
%% \end{equation}
\end{proof}

\subsection{Unique Multilinear Polynomial Representations}

A polynomial $P(X_1,\ldots,X_n)$ is called multilinear if it has
individual degree at most one on each variable.  Each function $f : \{
\pm 1 \}^n \rightarrow \mathbb{C}$ has a unique representation as a
multilinear polynomial in $\mathbb{C}[X_1,\ldots,X_n]$ given by the
Fourier or Walsh-Hadamard transform \cite{ODonnellBook}. Explicitly:
\begin{equation}
P_f(X_1,\ldots,X_n) = \sum_{S \subseteq [n]} \hat{f}(S) \prod_{i \in S} X_i,
\label{eqn:fourier}
\end{equation}
where
\begin{equation}
\hat{f}(S) = \frac{1}{2^n} \sum_{a \in \{\pm 1\}^n} f(a) \prod_{i \in S} a_i.
\end{equation}
The polynomial represents $f$ in the sense that $P_f(a) = f(a)$ holds
for every $a \in \{ \pm 1 \}^n$.  If the range of $f$ is a subset of
$\mathbb{R}$, then each $\hat{f}(S)$ is indeed a real number.  The
Convolution Formula describes the Fourier coefficients of pointwise
products $fg$ of functions $f,g : \{ \pm 1 \}^n \rightarrow
\mathbb{C}$.  It states that
\begin{equation}
\widehat{fg}(S) = \sum_{T \subseteq [n]} \hat{f}(S) \hat{g}(S \Delta T)
\label{eqn:convolution}
\end{equation}
for every $S \subseteq [n]$, where $S \Delta T$ denotes symmetric
difference; i.e.~$S \Delta T = (S \setminus T) \cup (T \setminus S)$.
% In this
% context, Parseval's Identity states that
% \begin{equation}
% \sum_{S \subseteq [n]} |\hat{f}(S)|^2 = \frac{1}{2^n} \sum_{a \in \{\pm
%   1\}^n} |f(a)|^2.
% \end{equation}

We give an example of use of the uniqueness of the Fourier transform
that will be useful later on. We begin by recalling some notation and
terminology. A \emph{literal} is a Boolean variable $x$ or its
negation $\neg x$.  The literals $x$ and $\neg x$ are said to be
\emph{complementary} of each other, and $x$ is their \emph{underlying
  variable}. If $\ell$ is a literal, then $\overline{\ell}$ denotes
its complementary literal. The \emph{sign} $\sg(\ell)$ of $\ell$ is
defined as follows: $\sg(\ell) = 1$ if $\ell = x$, and $\sg(\ell)= -1$
if $\ell = \neg x$, where $x$ is its underlying variable. Clearly,
$\sg(\overline{\ell}) = -\sg(\ell)$.

A \emph{clause} is a disjunction of literals. Let $C= (\ell_1 \lor
\cdots \lor \ell_r)$ be a clause.  In the $\pm 1$ representation of
Boolean values, the clause $C$ represents the relation $\{ \pm 1 \}^r
\setminus \{(\sg(\ell_1),\ldots,\sg(\ell_r))\}$, which will be denoted
by $R_C$.  The \emph{indicator} function of the clause $C= (\ell_1
\lor \cdots \lor \ell_r)$ is the Boolean function from $\{ \pm 1 \}^r
\to \{ \pm 1 \}$ that maps the tuple
$(\sg(\ell_1),\ldots,\sg(\ell_r))$ to $+1$ and every other tuple to
$-1$. We write $P_C(X_1, \ldots, X_r)$ to denote the unique
multilinear polynomial representation of the indicator function of
the clause $C$.

\begin{lemma} \label{lem:fourierofaclause}
Let $C = (\ell_1 \lor \cdots \lor \ell_r)$ be a clause on $r$
different variables.  Then, over the ring of polynomials
$\mathbb{C}[X_1,\ldots,X_r]$, the following identity holds.
\begin{equation}
P_C(X_1,\ldots,X_r) = 2^{1-r} \prod_{i=1}^r \Big({1+\sg(\ell_i)X_i}\Big) - 1.
\label{eqn:targeteqia}
\end{equation}
\end{lemma}

\begin{proof}
Let $R_C = \{ \pm 1 \}^r \setminus
\{(\sg(\ell_1),\ldots,\sg(\ell_r))\}$ be the Boolean relation
represented by $C$.  Since the right-hand side of
equation~\eqref{eqn:targeteqia} is a multilinear polynomial and its
left-hand side is the unique multilinear polynomial that agrees with
the indicator function of $R_C$ on $\{ \pm 1 \}$, it suffices to check
that the right-hand side also agrees with the indicator function of
$R_C$ on $\{ \pm 1 \}^r$. In other words, we claim that for every
$(a_1,\ldots,a_r) \in \{ \pm 1 \}^r$, the right-hand side evaluates to
$-1$ if the truth-assignment $(a_1,\ldots,a_r)$ satisfies the clause
$C$, and it evaluates to $1$, otherwise.

Assume that $(a_1,\ldots,a_r)$ satisfies the clause $C$. Then there is
some $j\in \{1,\ldots,r\}$ such that $a_j= -\sg(\ell_j)$. It follows
that $1+\sg(\ell_j)a_j= 0$ and so $\prod_{i=1}^r
\left({1+\sg(\ell_i)a_i}\right)= 0$, which, in turn, implies that
$P_C(a_1,\ldots,a_r) = -1$. Assume that $(a_1,\ldots,a_r)$ does not
satisfy the clause $C$. Then, for every $i \in \{1,\ldots,r\}$, we
have that $a_i= \sg(\ell_i)$. Consequently, for every
$i\in\{1,\ldots,r\}$, we have that $1+\sg(\ell_i)a_i= 2$ and so
$\prod_{i=1}^r \left({1+\sg(\ell_i)a_i}\right)= 2^r$, which, in turn,
implies that $P_C(a_1,\ldots,a_r) = 1$.  This completes the proof of
the claim.
\end{proof}

\subsection{Hilbert Space}

A Hilbert space is a complex vector space with an inner product whose
norm induces a complete metric. All Hilbert spaces of finite dimension
$d$ are isomorphic to $\mathbb{C}^d$ with the standard complex inner
product. In particular, this means that after the choice of a basis,
we can identify the linear operators on a $d$-dimensional Hilbert
space with the matrices in~$\mathbb{C}^{d \times d}$. Composition of
operators becomes matrix multiplication. A matrix $A$ is Hermitian if
it is equal to its conjugate transpose $A^*$. A diagonal matrix is one
all whose off-diagonal entries are $0$. A matrix $A$ in unitary if
$A^* A = A A^* = I$, where $I$ is the identity matrix. Two matrices
$A$ and $B$ commute if $AB = BA$, and a collection of matrices
$A_1,\ldots,A_r$ pairwise commute if $A_iA_j = A_jA_i$ for all $i,j
\in [r]$.

For the basics of general Hilbert spaces and their linear operators we
refer the reader to Halmos' monograph \cite{HalmosHilbertSpace}. We
need from it the concepts of bounded linear operator and of adjoint
$A^*$ of a densely defined linear operator~$A$. Two operators $A$ and
$B$ commute if $AB = BA$. A sequence of operators $A_1,\ldots,A_r$
pairwise commute if $A_iA_j = A_jA_i$ for all $i,j \in [r]$. A linear
operator $A$ is called normal if it commutes with its adjoint $A^*$;
i.e., $AA^* = A^*A$. A linear operator is called self-adjoint if $A^* =
A$. A linear map from a Hilbert space $\mathcal{H}_1$ to another
Hilbert space $\mathcal{H}_2$ is called unitary if it preserves norms.

We also make elementary use of general $L^2$- and
$L^\infty$-spaces. Let $(\Omega,\mathcal{M},\mu)$ be a measure
space. Then $L^2(\Omega,\mu)$ denotes the collection of square
integrable measurable functions, up to almost everywhere
equality.  Also $L^\infty(\Omega,\mu)$ denotes the collection of
essentially bounded measurable functions, up to almost
everywhere equality.  All measure-theoretic terms in these definitions
refer to $\mu$.  See \cite{FollandRealAnalysis} for definitions.

\subsection{Constraint Languages, Instances, Value and Satisfiability}

A \emph{Boolean constraint language} $A$ is a collection of relations
over the Boolean domain $\{\pm 1\}$. Let $V = \{X_1,\ldots,X_n\}$ be a
set of variables. An \emph{instance} $\mathcal{I}$ on the variables $V$
over the constraint language $A$ is a finite collection of pairs
\begin{equation}
\mathcal{I} = ((Z_1,R_1), \ldots, (Z_m,R_m))
\label{eqn:instance}
\end{equation}
where each $R_i$ is a relation from $A$ and $Z_i =
(Z_{i,1},\ldots,Z_{i,r_i})$ is a tuple of variables from $V$ or
constants from $\{ \pm 1 \}$, where $r_i$ is the arity of $R_i$. Each
pair $(Z_i,R_i)$ is called a \emph{constraint}, and each $Z_i$ is
called its \emph{constraint-scope}. A \emph{Boolean assignment} is a
mapping $f$ assigning a Boolean value $a_i \in \{ \pm 1 \}$ to each
variable~$X_i$, and assigning $-1$ and $+1$ to the constants $-1$ and
$+1$, respectively.  We say that the assignment satisfies the $i$-th
constraint if the tuple $f(Z_i) = (f(Z_{i,1}),\ldots,f(Z_{i,r_i}))$
belongs to $R_i$.  The \emph{value of $f$ on $\mathcal{I}$} is the
fraction of constraints that are satisfied by $f$. The \emph{value} of
$\mathcal{I}$, denoted by $\nu(\mathcal{I})$, is the maximum value
over all Boolean assignments.  We say that $\mathcal{I}$ is
\emph{satisfiable} in the Boolean domain if there is a Boolean
assignment that satisfies all constraints; equivalently, if
$\nu(\mathcal{I}) = 1$.

\subsection{Operator Assignments and Satisfiability via Operators}

Let $X_1,\ldots,X_n$ be $n$ variables, and let $\mathcal{H}$ be a
Hilbert space. An \emph{operator assignment} for $X_1,\ldots,X_n$ over
$\mathcal{H}$ is an assignment of a bounded linear operator on
$\mathcal{H}$ to each variable, $f : X_1,\ldots,X_n \mapsto
A_1,\ldots,A_n$, such that the following conditions hold:
\begin{enumerate} \itemsep=0pt
\item $A_j$ is self-adjoint for every $j \in [n]$,
\item $A_j^2 = I$ for every $j \in [n]$.
\end{enumerate}
If $S$ is a subset of $\{X_1,\ldots,X_n\}$, we say that the operator
assignment $A_1,\ldots,A_n$ \emph{pairwise commutes} on $S$ if in
addition it satisfies $A_j A_k = A_k A_j$ for every $X_j$ and $X_k$ in
the set $S$. If it pairwise commutes on the whole set
$\{X_1,\ldots,X_n\}$, we say that the assignment \emph{fully
  commutes}.

Let $A$ be a Boolean constraint language, let $\mathcal{I}$ be an
instance over $A$, with $n$ variables $X_1,\ldots,X_n$ as
in~\eqref{eqn:instance}, and let $\mathcal{H}$ be a Hilbert space. An
\emph{operator assignment for $\mathcal{I}$ over $\mathcal{H}$} is an
operator assignment $f : X_1,\ldots,X_n \mapsto A_1,\ldots,A_n$ for the
variables $X_1,\ldots,X_n$ that pairwise commutes on the set of
variables of each constraint scope $Z_i$ in $\mathcal{I}$; explicitly
\begin{equation}
A_j A_k = A_k A_j \;\;\;\; \text{ for every } X_j \text{ and }
X_k \text{ in } Z_i, \text{ for every } i \in \{1,\ldots,m\}.
\end{equation}
We also require that $f$ maps the constant $-1$ and $+1$ to $-I$ and
$I$, respectively, where $I$ is the identity operator on
$\mathcal{H}$.  We say that the assignment $f$ \emph{satisfies} the
$i$-th constraint if
\begin{equation}
P_{R_i}(f(Z_i)) = P_{R_i}(f(Z_{i,1}),\ldots,f(Z_{i,r_i})) = -I,
\end{equation}
where $P_{R_i}$ denotes the unique multilinear polynomial
representation of indicator function of the relation $R_i$, i.e., the
function that maps each tuple in $R_i$ to $-1$, and each tuple in its
complement $\{\pm 1\}^{r_i} \setminus R_i$ to $+1$.  Note that since
$f(Z_{i,1}),\ldots,f(Z_{i,r_i})$ are required to commute by
definition, this notation is unambiguous despite the fact that
$P_{R_i}$ is defined as a polynomial in commuting variables.  The
\emph{value} of $f$ on $\mathcal{I}$ is the fraction of constraints
that are satisfied by $f$; note that this quantity takes one of a
finite set of values in the set $\{0,1/m,2/m,\ldots,(m-1)/m,1\}$. The
\emph{value of $\mathcal{I}$ over $\mathcal{H}$} is the maximum value
over all operator assignments for $\mathcal{I}$ over $\mathcal{H}$. We
say that $f$ \emph{satisfies} $\mathcal{I}$ if it satisfies all
constraints.  In that case we also say that $f$ is a \emph{satisfying
  operator assignment} for $\mathcal{I}$ over $\mathcal{H}$.

The \emph{finite-dimensional value of $\mathcal{I}$}, denoted by
$\nu^*(\mathcal{I})$, is the maximum of its value over all
finite-dimensional Hilbert spaces. The \emph{value of $\mathcal{I}$},
denoted by $\nu^{**}(\mathcal{I})$, is the maximum of its value over
all Hilbert spaces. We say that an instance $\mathcal{I}$ is
\emph{satisfiable via finite-dimensional operator assignments}, or
satisfiable via fd-operators for short, if $\nu^*(\mathcal{I}) =
1$. We say that $\mathcal{I}$ is \emph{satisfiable via operator
  assignments}, or satisfiable via operators for short, if
$\nu^{**}(\mathcal{I}) = 1$.

%% file: section-3-spectral-theorem-long.tex
\section{The Strong Spectral Theorem}

The Spectral Theorem plays an important role in linear algebra and
functional analysis. It has also been used in the foundations of
quantum mechanics (for some recent uses see \cite{CleveM14,Ji2013}).
We will make a similar use of it, but we will also need the version of
this theorem for infinite-dimensional Hilbert spaces.  In this section
we discuss the statement, both for finite- and infinite-dimensional
Hilbert spaces, as well as one of its important applications that we
encapsulate in a lemma for later reuse.

\subsection{Statement}

In its most basic form, the Spectral Theorem for complex matrices
states that every Hermitian matrix is unitarily equivalent to a
diagonal matrix. Explicitly: if $A$ is a $d \times d$ Hermitian
matrix, then there exist a unitary matrix $U$ and a diagonal matrix
$E$ such that $A = U^{-1} E U$. In its strong form, the Strong
Spectral Theorem (SST) applies to sets of pairwise commuting Hermitian
matrices and is stated as follows.

\begin{theorem}[Strong Spectral Theorem; finite-dimensional case]
\label{thm:sstfd}
Let $A_1,\ldots,A_r$ be $d \times d$ Hermitian matrices, for some
positive integer $d$. If $A_1,\ldots,A_r$ pairwise commute, then there
exists a unitary matrix $U$ and diagonal matrices $E_1,\ldots,E_r$
such that $A_i = U^{-1} E_i U$ for every $i \in [r]$.
\end{theorem}

This form of the SST will be enough to discuss satisfiability via
fd-operators. For operator assignments over arbitrary Hilbert spaces,
we need to appeal to the most general form of the SST in which the
role of diagonal matrices is played by \emph{multiplication operators}
on an $L^2(\Omega,\mu)$-space. These are defined as follows.

Let $V$ be a complex function space; a complex vector space of
functions mapping indices from an index set $X$ to $\mathbb{C}$. A
multiplication operator of $V$ is a linear operator whose value at a
function $f : X \rightarrow \mathbb{C}$ in $V$ is given by pointwise
multiplication by a fixed function $a : X \rightarrow \mathbb{C}$. In
symbols, the multiplication operator given by $a$ is
\begin{equation}
(T_a(f))(x) = a(x)f(x) \;\;\;\;\; \text{ for each } x \in X.
\end{equation}
In its weak form, the general Spectral Theorem states that any normal
bounded linear operator on a Hilbert space is unitarily equivalent to
a multiplication operator on an $L^2$-space. We need the following
strong version of the Spectral Theorem that states that the same is
true for a collection of such operators, simultaneously through the
same unitary transformation, provided they commute.  The statement we
use is a direct consequence of Theorem~1.47 in Folland's monograph
\cite{FollandHarmonicAnalysis}.

\begin{theorem}[Strong Spectral Theorem; general case]
  \label{thm:sstgeneral}
  Let $A_1,\ldots,A_n$ be normal bounded linear operators on a Hilbert
  space $\mathcal{H}$.  If $A_1,\ldots,A_r$ pairwise commute, then
  there exist a measure space $(\Omega,\mathcal{M},\mu)$, a unitary
  map $U : \mathcal{H} \rightarrow L^2(\Omega,\mu)$, and functions
  $a_1,\ldots,a_r \in L^\infty(\Omega,\mu)$ such that $A_i = U^{-1}
  T_{a_i} U$ for every $i \in [r]$.
\end{theorem}

The special case in which $\mathcal{H}$ has finite dimension $d$, the
measure space is actually a finite set of cardinality $d$ with the
counting measure, and thus $L^2(\Omega,\mu)$ is isomorphic to
$\mathbb{C}^d$ with the usual complex inner product.

\subsection{An Oft-Used Application}

The following lemma encapsulates a frequently used application of the
Strong Spectral Theorem. It states that whenever a set of polynomial
equations entail another polynomial equation over the Boolean domain,
then the entailment holds as well for fully commuting operator
assignments.

\begin{lemma} \label{lem:entail} Let $X_1,\ldots,X_r$ be variables,
  let $Q_1,\ldots,Q_m,Q$ be polynomials in
  $\mathbb{C}[X_1,\ldots,X_r]$, and let $\mathcal{H}$ be a Hilbert
  space. If every Boolean assignment that satisfies the equations $Q_1
  = \cdots = Q_m = 0$ also satisfies the equation $Q = 0$, then every
  fully commuting operator assignment over $\mathcal{H}$ that
  satisfies the equations $Q_1 = \cdots = Q_m = 0$ also satisfies the
  equation~$Q = 0$.
\end{lemma}

Although the same proof applies to all Hilbert spaces,
the proof of the finite-dimensional case can be made
more elementary. Since for certain applications only the
finite-dimensional case of the lemma is relevant, we split the proof
accordingly into cases.

\begin{proof}[Proof of Lemma~\ref{lem:entail}; finite-dimensional case.]
Assume $\mathcal{H}$ has finite dimension $d$. Since all Hilbert
spaces of dimension $d$ are isometrically isomorphic to
$\mathbb{C}^d$, let us assume without loss of generality that
$\mathcal{H} = \mathbb{C}^d$. In such a case, a self-adjoint bounded
linear operator is just a Hermitian $d \times d$ matrix, and the
composition of linear operators is matrix multiplication.

Assume the hypotheses of the lemma and let $A_1,\ldots,A_r$ be
Hermitian $d \times d$ matrices. Assume that $A_1,\ldots,A_r$ make a
fully commuting operator assignment for $X_1,\ldots,X_r$ such that the
equations $Q_1 = \cdots = Q_m = 0$ are satisfied. The matrices
$A_1,\ldots,A_r$ pairwise commute, so the Strong Spectral Theorem
(i.e.\ Theorem~\ref{thm:sstfd}) applies to them. Thus, there exist a
unitary matrix $U$ and diagonal $d \times d$ matrices $E_1,\ldots,E_m$
such that $A_i = U^{-1} E_i U$ for every $i \in
[r]$. Equivalently, $UA_iU^{-1} = E_i$. From $A_i^2 = I$ we
conclude $E_i^2 = I$. Hence, if $a_i(j)$ denotes the $j$-th diagonal
entry of $E_i$, then $a_i(j)^2 = 1$ for all $j \in
[d]$. Thus $a_i(j) \in \{ \pm 1 \}$ for all $j \in
[d]$. The conditions of
Lemma~\ref{lem:simultaneouslysimilar} apply, so $Q_k(A_1,\ldots,A_r)$
and $Q_k(E_1,\ldots,E_r)$ are similar matrices for each $k \in
[m]$. Since $Q_k(A_1,\ldots,A_r) = 0$ and the unique matrix
that is similar to the null matrix is the null matrix itself, we
conclude that $Q_k(E_1,\ldots,E_r) = 0$. Now, $E_i$ is the diagonal
matrix that has the vector $(a_i(1),\ldots,a_i(d))$ in the diagonal,
so $Q_k(a_1(j),\ldots,a_r(j)) = 0$ for all $j \in
[d]$. Since $a_i(j)$ is in $\{ \pm 1 \}$ for each $i \in
[r]$ and $j \in [d]$, the hypothesis of the
lemma says that also $Q(a_1(j),\ldots,a_r(j)) = 0$ for all $j \in
[d]$. Thus $Q(E_1,\ldots,E_r) = 0$, and another application
of Lemma~\ref{lem:simultaneouslysimilar} shows that $Q(A_1,\ldots,A_r)
= 0$, as was to be proved. \end{proof}

The proof for the general case follows the same structure as the proof
of the finite-dimensional case, using Theorem~\ref{thm:sstgeneral} in
place of Theorem~\ref{thm:sstfd}. Other than taking care of \emph{null
  sets of exceptions}, there are no further differences in the two
proofs. At a later stage we will find an application of the SST whose
proof for the infinite-dimensional case does require some new
ingredients. For now, let us fill in the details of the
null-set-of-exceptions argument as a warm-up.

\begin{proof}[Proof of Lemma~\ref{lem:entail}; general case.]
Assume the hypotheses of the lemma and let $A_1,\ldots,A_r$ be bounded
self-adjoint linear operators on $\mathcal{H}$. Suppose that
$A_1,\ldots,A_r$ make a fully commuting operator assignment for
$X_1,\ldots,X_r$ such that the equations $Q_1 = \cdots = Q_m = 0$ are
satisfied. The operators $A_1,\ldots,A_r$ pairwise commute, and since
they are self-adjoint they are also normal, so the Strong Spectral
Theorem (i.e.\ Theorem~\ref{thm:sstgeneral}) applies to them. Thus,
there exist a measure space $(\Omega,\mathcal{M},\mu)$, a unitary map
$U : \mathcal{H} \rightarrow L^2(\Omega,\mu)$ and functions
$a_1,\ldots,a_r \in L^\infty(\Omega,\mu)$ such that, for the
multiplication operators $E_i = T_{a_i}$ of $L^2(\Omega,\mu)$, the
relations $A_i = U^{-1} E_i U$ hold for every $i \in
[r]$. Equivalently, $UA_iU^{-1} = E_i$. From $A_i^2 = I$ we conclude
$E_i^2 = I$. Hence, $a_i(\omega)^2 = 1$ for almost all $\omega \in
\Omega$; i.e.\ formally, $\mu(\{ \omega \in \Omega : a_i(\omega)^2
\not= 1 \}) = 0$. Thus $a_i(\omega) \in \{ \pm 1 \}$ for almost all
$\omega \in \Omega$. The conditions of
Lemma~\ref{lem:simultaneouslysimilar} apply, thus
$Q_k(A_1,\ldots,A_r)$ and $Q_k(E_1,\ldots,E_r)$ are similar linear
operators for each $k \in [m]$. Since $Q_k(A_1,\ldots,A_r)
= 0$ and the unique linear operator that is similar to the null
operator is the null operator itself, we conclude that
$Q_k(E_1,\ldots,E_r) = 0$. Now, $E_i$ is the multiplication operator
given by the function $a_i$, so $Q_k(a_1(\omega),\ldots,a_r(\omega)) =
0$ for almost all $\omega \in \Omega$.  Since for almost all $\omega
\in \Omega$ the component $a_i(\omega)$ is in $\{ \pm 1 \}$ for each
$i \in [r]$, the hypothesis of the lemma says that also
$Q(a_1(\omega),\ldots,a_r(\omega)) = 0$ for almost all $\omega \in
\Omega$. Thus $Q(E_1,\ldots,E_r) = 0$, and another application of
Lemma~\ref{lem:simultaneouslysimilar} shows that $Q(A_1,\ldots,A_r) =
0$, as was to be proved. \end{proof}

%% file: section-4-reductions-long.tex
\section{Reductions through Primitive Positive Formulas} \label{sec:pp}

Let $A$ be a Boolean constraint language, let $r$ be a positive
integer, and let $x_1,\ldots,x_r$ be variables ranging over the
Boolean domain $\{ \pm 1 \}$. A primitive positive formula, or
pp-formula for short, is a formula of the form
\begin{equation}
\phi(x_1,\ldots,x_r) = \exists y_1 \cdots \exists y_s
\left({R_1(z_1) \wedge \cdots \wedge R_m(z_m)} \right)
\end{equation}
where each $R_i$ is a relation in $A$ and each $z_i$ is an $r_i$-tuple
of variables or constants from $\{x_1,\ldots,x_r\} \cup
\{y_1,\ldots,y_s\} \cup \{ \pm 1 \}$, where $r_i$ is the arity of
$R_i$.  A relation $R \subseteq \{ \pm 1 \}^r$ is pp-definable from
$A$ if there exists a pp-formula $\phi(x_1,\ldots,x_r)$ such that
\begin{equation}
R = \{ (a_1,\ldots,a_r) \in \{ \pm 1 \}^r :
\phi(x_1/a_1,\ldots,x_r/a_r) \text{ is true in } A \}.
\end{equation}
A Boolean constraint language $A$ is pp-definable from another Boolean
constraint language $B$ if every relation in $A$ is pp-definable
from~$B$. Whenever the constants $+1$ and $-1$ do not appear in the
pp-formulas, we speak of pp-formulas and pp-definability \emph{without
  constants} or, also, \emph{without parameters}.

In the following we show that if $A$ is pp-definable from $B$, then
every instance $\mathcal{I}$ over $A$ can be translated into an
instance $\mathcal{J}$ over $B$ in such a way that the satisfying
operator assignments for $\mathcal{I}$ lift to satisfying operator
assignments for $\mathcal{J}$. We make this precise.

\subsection{The Basic Construction} \label{sec:basicconstruction}

Let $A$ and $B$ be two Boolean constraint languages and assume that
every relation in $A$ is pp-definable from $B$.  For $R$ in $A$, let
\begin{equation}
\phi_R(x_1,\ldots,x_r) = \exists y_1 \cdots \exists y_t (S_1(w_1)
\wedge \cdots \wedge S_m(w_m)) \label{eqn:ppformula}
\end{equation}
be the pp-formula that defines $R$ from $B$, where $S_1,\ldots,S_m$
are relations from $B$, and $w_1,\ldots,w_m$ are tuples of variables
or constants in $\{x_1,\ldots,x_r\} \cup \{y_1,\ldots,y_t\} \cup \{
\pm 1 \}$ of appropriate lengths. For every instance $\mathcal{I}$ of
$A$ we construct an instance $\mathcal{J}$ of $B$ as follows.

Consider a constraint $(Z,R)$ in $\mathcal{I}$, where $Z =
(Z_1,\ldots,Z_r)$ is a tuple of variables of $\mathcal{I}$ or
constants in $\{ \pm 1 \}$.  In addition to the variables in $Z$, in
$\mathcal{J}$ we add new fresh variables $Y_1,\ldots,Y_t$ for the
quantified variables $y_1,\ldots,y_t$ in $\phi_R$. We also add one
constraint $(W_j,S_j)$ for each $j \in [m]$, where $W_j$ is the tuple
of variables and constants obtained from $w_j$ by replacing the
variables in $x_1,\ldots,x_r$ by the corresponding components
$Z_1,\ldots,Z_r$ of $Z$, replacing any $y_i$-variable by the
corresponding $Y_i$, and leaving all constants untouched. We do this
for each constraint in $\mathcal{I}$ one by one. The collection of
variables $Z_1,\ldots,Z_r,Y_1,\ldots,Y_t$ that are introduced by the
constraint $(Z,R)$ of $\mathcal{I}$ is referred to as the \emph{block}
of $(Z,R)$ in $\mathcal{J}$. Note that two blocks of different
constraints may intersect, but only on the variables of $\mathcal{I}$.

This construction is referred to as a \emph{gadget reduction} in the
literature. Its main property for satisfiability in the Boolean domain
is the following straightforward fact:

\begin{lemma} \label{lem:ppclassical}
$\mathcal{I}$ is satisfiable in the Boolean domain if and only if
  $\mathcal{J}$ is.
\end{lemma}

\noindent We ommit its very easy proof. Our goal in the rest of this
section is to show that one direction of this basic property of gadget
reductions is also true for satisfiability via operators, for both
finite- and infinite-dimensional Hilbert spaces, and that the other
direction is \emph{almost true} in a sense we will make precise in due
time.

\subsection{Correctness: Operator Solutions Lift}

The following lemma shows that the left-to-right direction in
Lemma~\ref{lem:ppclassical} also holds for satisfiability via
operators: satisfying operator assignments for $\mathcal{I}$ can be
lifted to satisfying operator assignments for $\mathcal{J}$, over the
same Hilbert space.

\begin{lemma} \label{lem:liftonly} Let $\mathcal{I}$ and
  $\mathcal{J}$ be as above and let $\mathcal{H}$ be a Hilbert
  space. For every $f$ that is a satisfying operator assignment for
  $\mathcal{I}$ over $\mathcal{H}$, there exists $g$ that extends $f$
  and is a satisfying operator assignment for $\mathcal{J}$ over
  $\mathcal{H}$.  Moreover, $g$ is pairwise commuting on each block of
  $\mathcal{J}$.
\end{lemma}

As in the proof of Lemma~\ref{lem:entail} we split into cases.

\begin{proof}[Proof of Lemma~\ref{lem:liftonly}, finite-dimensional case]
  As in the proof of the finite-dimensional case of
  Lemma~\ref{lem:entail}, we may assume that $\mathcal{H} =
  \mathbb{C}^d$ for some positive integer $d$, and that
  $A_1,\ldots,A_n$ are Hermitian $d \times d$ matrices that make a
  satisfying operator assignment $f$ for $\mathcal{I}$. We need to
  define Hermitian matrices for the new variables of $\mathcal{J}$
  that were introduced by its construction. We define these matrices
  simultaneously for all variables $Y_1,\ldots,Y_t$ that come from the
  same constraint $(Z,R)$ of $\mathcal{I}$.

By renaming the entries in $Z$ if necessary, let us assume without
loss of generality that the variables in $Z$ are $X_1,\ldots,X_r$,
where $r$ is the arity of $R$. By the commutativity condition of
satisfying operator assignments, the matrices $A_1,\ldots,A_r$
pairwise commute. As each $A_i$ is Hermitian, the Strong Spectral
Theorem applies to them. Thus, there exist a unitary matrix $U$ and
diagonal $d \times d$ matrices $E_1,\ldots,E_r$ such that the
relations $A_i = U^{-1} E_i U$ hold for each $i \in
[r]$. Equivalently, $U A_i U^{-1} = E_i$. From $A_i^2 = I$ we conclude
$E_i^2 = I$. Hence, if $a_i(j)$ denotes the $j$-th diagonal entry of
$E_i$, then $a_i(j)^2 = 1$ for all $j \in [d]$.  Thus $a_i(j) \in \{
\pm 1 \}$ for all $j \in [d]$.  The conditions of
Lemma~\ref{lem:simultaneouslysimilar} apply, thus
$P_R(A_1,\ldots,A_r)$ and $P_R(E_1,\ldots,E_r)$ are similar
matrices. Since $P_R(A_1,\ldots,A_r) = -I$ and the unique matrix that
is similar to $-I$ is $-I$ itself, we conclude that
$P_R(E_1,\ldots,E_r) = -I$. Now, $E_i$ is the diagonal matrix that has
the vector $(a_i(1),\ldots,a_i(d))$ in the diagonal, so
$P_R(a_1(j),\ldots,a_r(j)) = -1$ for all $j \in [d]$. Thus the tuple
$a(j) = (a_1(j),\ldots,a_r(j))$ belongs to the relation $R$ for all $j
\in [d]$. Now we are ready to define the matrices for the variables
$Y_1,\ldots,Y_t$.

For each $j \in [d]$, let $b(j) = (b_1(j),\ldots,b_t(j)) \in \{ \pm 1
\}^t$ be a tuple of witnesses to the existentially quantified
variables in $\phi_R(x_1/a_1(j),\ldots,x_r/a_r(j))$; such a vector of
witnesses must exist since the tuple $a(j)$ belongs to $R$ and
$\phi_R$ defines $R$. Let $F_k$ be the diagonal matrix that has the
vector $(b_k(1),\ldots,b_k(d))$ in the diagonal, and let $Y_k$ be
assigned the matrix $B_k = U^{-1} F_k U$. Since $U$ is unitary, each
such matrix is Hermitian and squares to the identity since $b_k(j) \in
\{ \pm 1 \}$ for all $j \in [d]$. Moreover,
$E_1,\ldots,E_r,F_1,\ldots,F_t$ pairwise commute since they are
diagonal matrices; thus $A_1,\ldots,A_r,B_1,\ldots,B_t$ also pairwise
commute since they are simultaneously similar via $U$. Moreover, as
each atomic formula in the matrix of $\phi_R$ is satisfied by the
mapping sending $x_i \mapsto a_i(j)$ and $y_i \mapsto b_i(j)$ for all
$j \in [d]$, another application of
Lemma~\ref{lem:simultaneouslysimilar} shows that the matrices that are
assigned to the variables of this atomic formula make the
corresponding indicator polynomial evaluate to $-I$. This means that
the assignment to the $X$ and $Y$-variables makes a satisfying
operator assignment for the constraints of $\mathcal{J}$ that come from the
constraint $(Z,R)$ of $\mathcal{I}$. As different constraints from
$\mathcal{I}$ produce their own sets of $Y$-variables, these
definitions of assignments do not conflict with one another, and the
proof of the lemma is complete.
\end{proof}

The proof for the general case
requires some new ingredients. Besides the need to take care of null
sets of exceptions as in the proof of Lemma~\ref{lem:entail}, a new
complication arises from the need to build the operators for the new
variables that are introduced by the reduction. Concretely we need to
make sure that the \emph{functions} of witnesses, in contraposition to
the \emph{finite tuples} of witnesses in the finite-dimensional case,
are bounded and measurable. We go carefully through the argument.

\begin{proof}[Proof of Lemma~\ref{lem:liftonly}, general case]
Assume that $A_1,\ldots,A_n$ are bounded self-adjoint linear operators
on $\mathcal{H}$ for the variables of $\mathcal{I}$. Suppose that the
operators $A_1,\ldots,A_r$ make a valid satisfying operator assignment
for $\mathcal{I}$.  We need to define bounded self-adjoint linear
operators for the new variables of $\mathcal{J}$ that were introduced
by the construction. We define these operators simultaneously for all
variables $Y_1,\ldots,Y_t$ that come from the same constraint $(Z,R)$
of $\mathcal{I}$.

By renaming the components of $Z$ if necessary, assume without loss of
generality that the variables in $Z$ are $X_1,\ldots,X_r$, where $r$
is the arity of $R$. By the commutativity condition of satisfying
operator assignments, the operators $A_1,\ldots,A_r$ pairwise
commute. As each $A_i$ is self-adjoint, it is also normal, and the
Strong Spectral Theorem (c.f.\ Theorem~\ref{thm:sstgeneral})
applies. Thus, there exist a measure space $(\Omega,\mathcal{M},\mu)$,
a unitary map $U : \mathcal{H} \rightarrow L^2(\Omega,\mu)$ and
functions $a_1,\ldots,a_r \in L^\infty(\Omega,\mu)$ such that, for the
multiplication operators $E_i = T_{a_i}$ of $L^2(\Omega,\mu)$, the
relations $A_i = U^{-1} E_i U$ hold for each $i \in
[r]$. Equivalently, $U A_i U^{-1} = E_i$. From $A_i^2 = I$ we conclude
$E_i^2 = I$. Hence, $a_i(\omega)^2 = 1$ for almost all $\omega \in
\Omega$; i.e., formally $\mu(\{ \omega \in \Omega : a_i(\omega)^2
\not= 1 \}) = 0$. Thus, $a_i(\omega) \in \{ \pm 1 \}$ for almost all
$\omega \in \Omega$. The conditions of
Lemma~\ref{lem:simultaneouslysimilar} apply, thus
$P_R(A_1,\ldots,A_r)$ and $P_R(E_1,\ldots,E_r)$ are similar linear
operators. Since $P_R(A_1,\ldots,A_r) = -I$ and the unique linear
operator that is similar to $-I$ is $-I$ itself, we conclude that
$P_R(E_1,\ldots,E_r) = -I$. Now, $E_i$ is the multiplication operator
given by $a_i$, and $a_i(\omega) \in \{ \pm 1 \}$ for almost all
$\omega \in \Omega$, so $P_R(a_1(\omega),\ldots,a_r(\omega)) = -1$ for
almost all $\omega \in \Omega$. Thus the tuple $a(\omega) =
(a_1(\omega),\ldots,a_r(\omega))$ belongs to the relation $R$ for
almost all $\omega \in \Omega$. Now we are ready to define the
operators for the variables~$Y_1,\ldots,Y_t$.

For each $\omega \in \Omega$ for which the tuple $a(\omega)$ belongs
to $R$, let $b(\omega) = (b_{1}(\omega),\ldots,b_{t}(\omega)) \in\{\pm
1\}^t$ be the \emph{lexicographically smallest} tuple of witnesses to
the existentially quantified variables in
$\phi_{R}(x_1/a_1(\omega),\ldots,x_r/a_r(\omega))$; such a vector of
witnesses must exist since $\phi_R$ defines~$R$, and the
lexicographically smallest exists because $R$ is finite. For every
other $\omega \in \Omega$, define $b(\omega) =
(b_1(\omega),\ldots,b_t(\omega)) = (0,\ldots,0)$.

Note that each function $b_k : \Omega \rightarrow \mathbb{C}$ is
bounded since its range is in $\{-1,0,1\}$. We claim that such
functions of witnesses $b_k$ are also measurable functions of
$(\Omega,\mathcal{M},\mu)$. This will follow from the fact that
$a_1,\ldots,a_r$ are measurable functions themselves, the fact
that~$R$ is a finite relation, and the choice of a definite tuple of
witnesses of each $\omega \in \Omega$; the lexicographically smallest
if $a(\omega)$ is in $R$, or the all-zero tuple otherwise. We discuss
the details.

Since $R$ is finite, the event $Q = \{ \omega \in \Omega : b_k(\omega)
= \sigma \}$, for fixed $\sigma \in \{ +1,0,-1 \}$, can be expressed
as a finite Boolean combination of events of the form $Q_{i,\tau} = \{
\omega \in \Omega : a_i(\omega) = \tau \}$, where $i \in [r]$ and
$\tau \in \{ \pm 1 \}$. Here is how: If $\sigma \not= 0$,
then
\begin{equation}
Q = \bigcup_{\substack{a \in R:\\b(a)_k = \sigma}} \Big( \bigcap_{i \in [r]}
Q_{i,a_i} \Big),
\end{equation}
where $b(a)$ denotes the lexicographically smallest tuple of witnesses
in $\{ \pm 1 \}^t$ for the quantified variables in
$\phi_R(x_1/a_1,\ldots,x_r/a_r)$. If $\sigma = 0$, then $Q$ is the
complement of this set. Each $Q_{i,\tau}$ is a measurable
set in the measure space $(\Omega,\mathcal{M},\mu)$ since $a_i$ is a
measurable function and $Q_{i,\tau} = a_i^{-1}(B_{1/4}(\tau))$, where
$B_{1/4}(\tau)$ denotes the complex open ball of radius $1/4$ centered
at $\tau$, which is a Borel set in the standard topology
of~$\mathbb{C}$. Since the range of $b_k$ is in the
finite set $\{-1,0,1\}$, the preimage $b_k^{-1}(S)$ of each Borel
subset $S$ of $\mathbb{C}$ is expressed as a finite Boolean
combination of measurable sets, and is thus measurable in
$(\Omega,\mathcal{M},\mu)$.

We just proved that each $b_k$ is bounded and measurable, so its
equivalence class under almost everywhere equality is represented in
$L^\infty(\Omega,\mu)$. We may assume without loss of generality that
$b_k$ is its own representative; else modify it on a set of measure
zero in order to achieve so. Let $F_k = T_{b_k}$ be the multiplication
operator given by $b_k$ and let $Y_k$ be assigned the linear operator
$B_k = U^{-1} F_k U$, which is bounded because $b_k$ is bounded and
$U$ is unitary. Also because $U$ is unitary, each such operator is
self-adjoint and squares to the identity since $b_{k}(\omega) \in \{
\pm 1 \}$ for almost all $\omega \in \Omega$. Moreover,
$E_1,\ldots,E_r,F_1,\ldots,F_t$ pairwise commute since they are
multiplication operators; thus $A_1,\ldots,A_r,B_1,\ldots,B_t$
pairwise commute since they are simultaneously similar via
$U$. Moreover, as each atomic formula in the matrix of $\phi_{R}$ is
satisfied by the mapping sending $x_i \mapsto a_{i}(\omega)$ and $y_i
\mapsto b_i(\omega)$ for almost all $\omega \in \Omega$, another
application of~Lemma~\ref{lem:simultaneouslysimilar} shows that the
operators that are assigned to the variables of this atomic formula
make the corresponding indicator polynomial evaluate to~$-I$. This
means that the assignment to the $X$ and $Y$-variables makes a
satisfying operator assigment for the constraints of $\mathcal{J}$
that come from the constraint $(Z,R)$ in $\mathcal{I}$. As different
constraints from $\mathcal{I}$ produce their own sets of
$Y$-variables, these definitions of assignments are not in conflict
with each other, and the proof of the lemma is complete.
\end{proof}

\subsection{The Extended Construction} \label{sec:extendedconstruction}

We proved so far that satisfying operator assignments for
$\mathcal{I}$ lift to satisfying operator assignments for
$\mathcal{J}$. We do not know if the converse is true. One could try
to just take the restriction of the satisfying assignment for
$\mathcal{J}$ to the variables of $\mathcal{I}$, but there is little
chance that this will work because there is no guarantee that the
operators that are assigned to any two variables that appear together
in a constraint of $\mathcal{I}$ will commute.  Instead of trying to
modify the assignment, we modify the instance $\mathcal{J}$. Let us
discuss a slightly modified version of $\mathcal{J}$, over a very
minor extension of the constraint language $B$, that still allows
lifting of solutions, and for which the naif projection works for the
backward direction. Let us stress now that we plan to use this
modified construction over a minor extension of the constraint
language merely as a technical device to get other results.

In the following, let $\mathrm{T}$ denote the full binary Boolean
relation; i.e., $\mathrm{T} = \{ \pm 1 \}^2$. Observe that the
indicator polynomial $P_{\mathrm{T}}(X_1,X_2)$ of the relation
$\mathrm{T}$ is just the constant $-1$; the letter~$\mathrm{T}$ stands
for \emph{true}.

Let $A$ and $B$ be the constraint languages such that $A$ is
pp-definable from $B$. Let $\mathcal{I}$ and $\mathcal{J}$ be the
instances over $A$ and $B$ as defined above. The modified version of
$\mathcal{J}$ will be an instance over the expanded constraint
language $B \cup \{\mathrm{T}\}$. We denote it $\mathcal{\hat{J}}$ and
it is defined as follows: the variables and the constraints of
$\mathcal{\hat{J}}$ are defined as in $\mathcal{J}$, but we also add
all the binary constraints of the form $((X_i,X_j),\mathrm{T})$,
$((X_i,Y_k),\mathrm{T})$ or $((Y_k,Y_\ell),\mathrm{T})$, for every
four different variables $X_i$, $X_j$, $Y_k$ and $Y_\ell$ that come
from the same block in $\mathcal{J}$.

\subsection{Correctness: Operator Solutions Lift and also Project}

We argue that in this new construction, satisfying assignments not
only lift from $\mathcal{I}$ to $\mathcal{\hat{J}}$, but also project
from $\mathcal{\hat{J}}$ to $\mathcal{I}$.

\begin{lemma} \label{lem:liftandproject} Let $\mathcal{I}$ and
  $\mathcal{\hat{J}}$ be as above and let $\mathcal{H}$ be a Hilbert
  space. Then the following assertions are true.
\begin{enumerate} \itemsep=0pt
\item For every $f$ that is a satisfying operator assignment for
  $\mathcal{I}$ over $\mathcal{H}$, there exists $g$ that extends $f$
  and is a satisfying operator assignment for $\mathcal{\hat{J}}$ over
  $\mathcal{H}$,
\item For every $g$ that is a satisfying operator assignment for
  $\mathcal{\hat{J}}$ over $\mathcal{H}$, the restriction $f$ of $g$
  to the variables of $\mathcal{I}$ is a satisfying operator assignment
  for $\mathcal{I}$ over $\mathcal{H}$.
\end{enumerate}
\end{lemma}

\begin{proof}
  Statement \emph{1} follows from Lemma~\ref{lem:liftonly}: Fix $f$
  that is a satisfying operator assignment for $\mathcal{I}$ and let
  $g$ be given by Lemma~\ref{lem:liftonly}. This is also an assignment
  for the variables of $\mathcal{\hat{J}}$. The constraints of
  $\mathcal{\hat{J}}$ that are already in $\mathcal{J}$ are of course
  satisfied by $g$. Next consider an additional constraint of the form
  $((X_i,X_j),\mathrm{T})$, $((X_i,Y_k),\mathrm{T})$ or
  $((Y_k,Y_\ell),\mathrm{T})$, for variables $X_i$, $X_j$, $Y_k$ and
  $Y_\ell$ coming from the same block in $\mathcal{J}$. By the
  ``moreover'' clause in Lemma~\ref{lem:liftonly}, the operators
  $A_i$, $A_j$, $B_k$ and $B_\ell$ associated to $X_i$, $X_j$, $Y_k$
  and $Y_\ell$ by $g$ pairwise commute. Moreover, the associated
  polynomial constraints $P_{\mathrm{T}}(A_i,A_j) = -I$,
  $P_{\mathrm{T}}(A_i,B_k) = -I$ and $P_{\mathrm{T}}(B_k,B_\ell) = -I$
  are trivial (i.e., void) since the indicator polynomial
  $P_{\mathrm{T}}(X_1,X_2)$ of $\mathrm{T}$ is just the constant $-1$.

  For statement~\emph{2}, fix $g$ that is a satisfying operator
  assignment for $\mathcal{\hat{J}}$ over $\mathcal{H}$, and let $f$
  be the restriction of $g$ to the variables of $\mathcal{I}$. Since
  $g$ satisfies $\mathcal{\hat{J}}$, for every two variables $X_i$ and
  $X_j$ that appear together in a constraint $(Z,R)$ of $\mathcal{I}$,
  the associated operators $g(X_i)$ and $g(X_j)$ commute since $X_i$
  and $X_j$ appear in the same block of $\mathcal{J}$. Hence $f(X_i)$
  and $f(X_j)$ commute. We still need to show that the polynomial
  constraint $P_{\mathrm{R}}(f(Z)) = -I$ is satisfied for every
  constraint $(Z,R)$ of $\mathcal{I}$. To do so, we use
  Lemma~\ref{lem:entail} on an appropriately defined system of
  polynomial equations.

  Let $r$ be the arity of $R$ and let $\phi_R$ be the pp-formula as
  in~\eqref{eqn:ppformula} that defines $R$ from $B$. The polynomials
  we define have variables
  $X_1,\ldots,X_r,Y_1,\ldots,Y_t,Z_{-1},Z_{+1}$ that correspond to the
  variables and constants in~\eqref{eqn:ppformula}. For every $k \in
  [m]$, let $Q_k$ be the polynomial $P_{S_k}(W_k) + 1$, so that the
  equation $Q_k = 0$ ensures $P_{S_k}(W_k) = -1$, where $P_{S_k}$ is
  the characteristic polynomial of $S_k$, and $W_k$ is the tuple of
  components from $X_1,\ldots,X_r,Y_1,\ldots,Y_s,Z_{-1},Z_{+1}$ that
  appear in the atom $S_k(w_k)$ of~\eqref{eqn:ppformula}. Here we use
  $X_i$ and $Y_j$ in place of $x_i$ and $y_j$, respectively, and
  $Z_{-1}$ and $Z_{+1}$ in place of the constants $-1$ and $+1$,
  respectively. Let also $Q_{m+1}$ and $Q_{m+2}$ be the polynomials
  $Z_{-1} + 1$ and $Z_{+1} - 1$, so that the equations $Q_{m+1} =
  Q_{m+2} = 0$ ensure that $Z_{-1} = -1$ and $Z_{+1} = +1$. Finally,
  let $Q$ be the polynomial $P_R(X_1,\ldots,X_r) + 1$, where $P_R$ is
  the characteristic polynomial of $R$. It follows from the
  definitions that every Boolean assignment that satisfies all
  equations $Q_1 = \cdots = Q_{m+2} = 0$ also satisfies $Q = 0$. Thus
  Lemma~\ref{lem:entail} applies, and since $g$ extended to $g(Z_{-1})
  = -I$ and $g(Z_{+1}) = I$ satisfies all equations $Q_1 = \cdots =
  Q_{m+2} = 0$, it also satisfies $Q = 0$. It follows that $P_R(f(Z))
  = P_R(g(Z)) = -I$, as was to be proved.
\end{proof}

%% file: section-5-gaps-long.tex
\section{Satisfiability Gaps via Operator Assignments} \label{sec:gaps}

Let $A$ be a Boolean constraint language and let $\mathcal{I}$ be an
instance over $A$. It is easy to see that the following inequalities hold:
\begin{equation}
  \nu(\mathcal{I}) ~ \leq   \nu^*(\mathcal{I}) ~ \leq \nu^{**}(\mathcal{I}).
  \label{eqn:directrelationships}
\end{equation}
Indeed, the first inequality holds because if we interpret the
field of complex numbers $\mathbb{C}$ as a 1-dimensional Hilbert
space, then the only solutions to the equation $X^2 = 1$ are $X = -1$
and $X = +1$. The second inequality is a direct consequence of the
definitions.  For the same reason, if $\mathcal{I}$ is
satisfiable in the Boolean domain, then it is satisfiable via fd-operators, and if it is satisfiable via fd-operators, then
it is satisfiable via operators.  The converses are, in
general, not true; however, finding counterexamples is
a non-trivial task. For the Boolean constraint language LIN of affine relations, counterexamples are given by Mermin's magic square
 \cite{Mermin1990,Mermin1993} for the first case, and by Slofstra's  recent construction \cite{Slofstra2016} for the second case. These will be discussed
 at some length in due time. In the rest of this
section, we characterize the Boolean constraint languages that exhibit such gaps.

We distinguish three types of gaps. Specifically, we say
that an instance $\mathcal{I}$ witnesses
% \begin{enumerate} \itemsep=0pt
% \item a quantum gap of the 1st kind if $\nu(\mathcal{I}) <
%   \nu^*(\mathcal{I})$,
% \item a quantum gap of the 2nd kind if $\nu(\mathcal{I}) <
%   \nu^{**}(\mathcal{I})$,
% \item a quantum gap of the 3rd kind if $\nu^*(\mathcal{I}) <
%   \nu^{**}(\mathcal{I})$,
% \end{enumerate}
% and
\begin{enumerate} \itemsep=0pt
\item a  \emph{satisfiability gap of the first kind} if
$\nu(\mathcal{I}) < 1$ and $\nu^*(\mathcal{I}) = 1$;
\item a \emph{satisfiability gap of the second kind} if
$\nu(\mathcal{I}) < 1$ and $\nu^{**}(\mathcal{I}) = 1$;
\item a \emph{satisfiability gap of the third kind} if
$\nu^*(\mathcal{I}) < 1$ and $\nu^{**}(\mathcal{I}) = 1$.
\end{enumerate}
As a mnemonic rule, count the number of stars $^*$ that appear in the
defining inequalities in 1,~2 or~3 to recall what kind the gap is.

We say that a Boolean constraint language $A$ has a \emph{satisfiability gap of the $i$-th kind}, $i= 1,2,3$,  if
there is at least one instance $\mathcal I$ over $A$ that witnesses such a gap.  Clearly,
a gap of the first kind or a gap of the third kind implies a gap of the second kind.  In other words, if $A$ has no gap of the second kind, then $A$ has no gap of the first kind and no gap of the third kind. A priori
 no other relationships seem to hold. %We show that, surprisingly, if a
%constraint language has a gap of one kind, then  it also has a gap of
%the other two kinds. 
We show that, in a precise sense, either $A$ has no gaps of any kind or $A$ has a gap of every kind.
 Recall from
Section~\ref{sec:pp} that $\mathrm{T}$ denotes the full binary Boolean
relation; i.e.\ $\mathrm{T} = \{ \pm 1 \}^2$. We are now ready to state and prove the main result of this section.

\begin{theorem} \label{thm:gaps}
Let $A$ be a Boolean constraint language. Then the following statements are equivalent.
\begin{enumerate} \itemsep=0pt
\item $A$ does not have a satisfiability gap of the first kind.
\item $A$ does not have a satisfiability gap of the second kind.
\item $A \cup \{ \mathrm{T} \}$ does not have a
  satisfiability gap of the third kind,
\item $A$ is $0$-valid, or $A$ is $1$-valid, or $A$ is bijunctive, or $A$ is Horn, or $A$ is dual Horn.
\end{enumerate}
\end{theorem}

%As we pointed out already, the implication~\emph{2} to~\emph{1}
%follows directly from equation~\eqref{eqn:directrelationships} but the
%rest of implications require proofs. First we prove each of the
%implications~\emph{1} to~\emph{4}, \emph{2} to~\emph{4}, and \emph{3}
%to~\emph{4} by exploiting Post's lattice and analysing the effect of
%pp-definability for the presence of gaps. Then we argue that the
%implications~\emph{4} to~\emph{1}, \emph{4} to~\emph{2}, and \emph{4}
%to~\emph{3} hold by strengthening a related result of Ji to make it
%work also for infinite-dimensional Hilbert spaces.

The proof of Theorem \ref{thm:gaps} has two main parts. In the first part, we show that if $A$ satisfies at least one of the conditions in the fourth statement, then $A$ has no satisfiability gaps of the first kind or the second kind, and $A\cup\{{\rm T}\}$ has no satisfiability gaps of the third kind. In the second part, we show that, in all other cases,  $A$ has satisfiability  gaps of the first kind and the second kind, and $A \cup \{{\rm T}\}$ has satisfiability gaps of the third kind. The ingredients in the proof of the second part are the existence of gaps of all three kinds for LIN, results about Post's lattice \cite{Post1941}, and \emph{gap-preserving} reductions that use the results about pp-definability established in Section \ref{sec:pp}.

\subsection{No Gaps of Any Kind} \label{subsec:no-gaps}

Assume that $A$ satisfies at least one of the conditions in the fourth statement in Theorem~\ref{thm:gaps}. 
First, we observe that the full relation ${\rm T}$ is $0$-valid, $1$-valid, bijunctive, Horn, and dual Horn. Indeed, ${\rm T}$ is obviously $0$-valid and $1$-valid. Moreover, it is bijunctive, Horn, and dual Horn because it is equal to the set of satisfying assignments of the Boolean formula $(x\lor \neg x)\land (y \lor \neg y)$, which is bijunctive, Horn, and dual Horn. Therefore, to prove that the fourth statement in Theorem~\ref{thm:gaps} implies the other three statement, it suffices to prove that if $A$ satisfies at least one of the conditions in the fourth statement, then $A$ has no gaps of any kind. Towards this goal, we argue by cases.

We start with the trivial cases in which  $A$ is $0$-valid or $1$-valid. If an instance $\mathcal{I}$ of $A$ contains a constraint of the form $(Z,\mathrm{F})$,
where $\mathrm{F}$ is an empty relation (of some arity), then $\mathcal{I}$ is not satisfiable by any operator assignment. Otherwise, $\mathcal{I}$ is satisfiable in the Boolean domain, hence it is satisfiable by assigning the identity operator $I$ to every variable, if $A$ is $0$-valid, or by assigning the operator $-I$ to every variable, if $A$ is $1$-valid.

Next, we have to show that if $A$ is bijunctive or Horn or dual Horn, then $A$ has no gaps of any kind. As discussed earlier, it suffices to show that $A$ does not have a gap of the second kind (since a gap of the first kind or a gap of the third kind implies a gap of the second kind).

  Ji \cite{Ji2013} proved that if $\mathcal{I}$ is a 2SAT instance  or a HORN SAT
instance that is  satisfiable via fd-operators, then $\mathcal{I}$  is also
satisfiable in the Boolean domain. In other words, Ji showed that 2SAT and  HORN~SAT have no gaps of the first kind.
This is quite close to what we have to prove, but there are two differences.
 First, a constraint language $A$ of Boolean relations is bijunctive  if every  relation in $A$ is the set of satisfying assignments of a 2CNF-formula,  but this formula need not be a $2$-clause. Similarly, $A$ is Horn (dual Horn) if every relation in $A$ is the set of satisfying assignments of  a Horn (dual Horn) formula, but this formula need not be a Horn (dual Horn) clause. This, however, is a minor complication that can be handled with some additional arguments, the details of which will be provided later on.
Second, at first glance, Ji's proof for 2SAT and  HORN SAT does not seem to extend to operator assignments of arbitrary (finite or infinite) dimension.
 The reason for this is that Ji's argument relies on the
existence of eigenvalues and associated orthogonal eigenspaces for the
linear operators, which are not guaranteed to exist in the
infinite-dimensional case, even for self-adjoint bounded linear
operators. Note however that in our case we have the additional
requirement that the operators satisfy $A^2 = I$, and in such a case
their eigenvalues and associated eigenspaces can be reinstated. This observation could perhaps be used to give a proof along the lines of Ji's
that 2SAT and  HORN SAT have no gaps of the second kind. However, we prefer
to give an alternative and more direct proof  that does
not rely at all  on the existence of eigenvalues. Our proof is based on
the manipulation of non-commutative polynomial identities, a method
that has been called \emph{the substitution method} (see,
e.g., \cite{CleveM14}).

\begin{lemma} \label{lem:twosatandhornsat} Let $\mathcal{I}$ be a {\rm 2SAT} instance or a {\rm HORN SAT} instance or a {\rm DUAL~HORN~SAT} instance. Then the following statements are equivalent.
  \begin{enumerate} \itemsep=0pt
  \item $\mathcal{I}$ is  satisfiable in the Boolean domain;
  \item $\mathcal{I}$ is  satisfiable via fd-operators;
  \item $\mathcal{I}$ is  satisfiable via operators.
  \end{enumerate}
\end{lemma}

\noindent We split the proof into  two: one for 2SAT and
another one for HORN SAT; the proof for  DUAL HORN SAT is analogous to the proof for HORN SAT, and it is omitted.

\begin{proof}[Proof of Lemma~\ref{lem:twosatandhornsat} for 2SAT]
Let $\mathcal{I}$ be a 2CNF-formula.  The implications $\emph{1}
\Longrightarrow \emph{2}$ and $ \emph{2} \Longrightarrow \emph{3}$
follow from the definitions.  To prove the implication $\emph{3}
\Longrightarrow \emph{1}$, assume that $f$ is a satisfying operator
assigment for $\mathcal{I}$ over a (finite-dimensional or infinite-dimensional)
Hilbert space $\mathcal{H}$, and, towards a contradiction, assume that
$\mathcal{I}$ is  unsatisfiable in the Boolean domain. We will make use of the
well-known characterization of unsatisfiable in the Boolean domain 2SAT
instances in terms of a reachability property of their associated
\emph{implication graph}. For~$\mathcal{I}$, the implication graph is
the directed graph $G$ that has one vertex for each literal $x$ or
$\neg x$ of every variable $x$ in $\mathcal{I}$, and two directed
edges for each clause $(\ell_1 \vee \ell_2)$ of $\mathcal{I}$, one
edge from $\overline{\ell_1}$ to $\ell_2$, and another one from
$\overline{\ell_2}$ to $\ell_1$. The well-known characterization
states that $\mathcal{I}$ is unsatisfiable in the Boolean domain if and only if there exists
a variable $x$ and two directed paths in $G$, one from the variable
$x$ to the literal $\neg x$, and another one from the literal $\neg x$
to the variable $x$ (see, e.g., \cite{Papadimitriou94}). Accordingly,
let $\ell_1,\ldots,\ell_r$ and $m_1,\ldots, m_s$ be literals such that
$x,\ell_1,\ldots,\ell_r,\neg x$ and $\neg x,m_1,\ldots,m_s,x$ are the
vertices in the paths from $x$ to $\neg x$ and from $\neg x$ to $x$,
respectively, in the order they are traversed.

  The existence of the path $x,\ell_1,\ldots,\ell_r,\neg x$ from the
  variable $x$ to the literal $\neg x$ in the implication graph $G$
  means that the clauses
  \begin{equation}
    (\neg x\lor \ell_1),\ (\overline{\ell_1}\lor \ell_2),\ \ldots, (\overline{\ell_{r-1}}\lor \ell_r),\ (\overline{\ell_r}\lor \neg x)
  \end{equation}
  are clauses of the instance $\mathcal{I}$. Symmetrically, the
  existence of the path $\neg x,m_1,\ldots,m_s, x$ from the literal
  $\neg x$ to the variable $ x$ in the implication graph
  $G$ means that the clauses
  \begin{equation}
    ( x\lor m_1),\ (\overline{m_1}\lor m_2),\ \ldots, (\overline{m_{s-1}}\lor m_s),\ (\overline{m_s}\lor  x)
  \end{equation}
  are clauses of the instance $\mathcal{I}$.

  In the case of satisfiability in the Boolean domain, one reasons that the
  instance $\mathcal{I}$ is unsatisfiable, because if it were
  satisfiable by some truth assignment, then the path of implications
  from $x$ to $\neg x$ forces $x$ to be set to \emph{false}, while the
  path of implications from $\neg x$ to $x$ forces $x$ to be set to
  \emph{true}. In what follows, we will show that, with some care,
  essentially the same reasoning can be carried out for operator
  assignments that satisfy the instance $\mathcal{I}$.

  Extend the operator assignment $f$ to all literals by setting
  $f(\ell) = \sg(\ell)f(x)$, where $x$ the variable underlying $\ell$.
  Since $f$ is a quantum satisfying assignment for $\mathcal{I}$,
  Lemma~\ref{lem:fourierofaclause} implies that
  \begin{eqnarray}
  (I-f(x))(I+f(\ell_1)) &=& 0 \label{eq:base}\\
   (I-f(\ell_i))(I+f(\ell_{i+1})) & = & 0, \quad 1\leq i\leq r-1.  \label{eq:i}\\
   (I-f(\ell_r))(I-f(x))   & = & 0 \label{eq:r}
   %\\
   %\cdots\\
   %  (1-\sg(\ell_r)f(\ell_r))(1-f(x))=0).
  \end{eqnarray}
  We now claim that
  \begin{eqnarray} \label{eq:claim}
  (I-f(x))(I+f(\ell_i)) & = &0,  \quad 1\leq i\leq r.
  \end{eqnarray}

  We prove the claim by induction on $i$.  For $i = 1$, what we need
  is just equation~\eqref{eq:base}.  By induction, assume now that
    \begin{eqnarray} \label{eq:ind}
    (I-f(x))(I+f(\ell_{i-1})) & = &  0
    \end{eqnarray}
     holds for some $i$ with $2\leq i\leq r-1$. By~\eqref{eq:i}, we
     have that
     \begin{eqnarray} \label{eq:i-1}
     (I-f(\ell_{i-1}))(I+f(\ell_{i})) & = & 0.
     \end{eqnarray}
    holds.  First, by multiplying equation~\eqref{eq:ind} from the
    right by $(I+f(\ell_i))$, we get
  \begin{eqnarray} \label{eqn:first-prod}
 (I-f(x))(1+f(\ell_{i-1}))(1+f(\ell_i)) & = &  0
  \end{eqnarray}
  Second, by multiplying equation~\eqref{eq:i-1} from the left by
  $(I-f(x))$, we get
   \begin{eqnarray} \label{eqn:second-prod}
 (I-f(x))(1-f(\ell_{i-1}))(1+f(\ell_i)) & = & 0
  \end{eqnarray}
 By adding equations~\eqref{eqn:first-prod}
 and~\eqref{eqn:second-prod}, we obtain
  \begin{eqnarray} \label{eqn:goal}
 (I-f(x))(I+f(\ell_i)) & = &  0,
  \end{eqnarray}
  as desired. In particular, by considering the case $i=r$, we get
 \begin{eqnarray}
 (I-f(x))(I+f(\ell_r)) & = &  0,
  \end{eqnarray}
which, after multiplying out the left-hand side, becomes
\begin{eqnarray} \label{eqn:goal2}
 I+ f(\ell_r) -f(x) - f(x) f(\ell_r) & = &  0.
  \end{eqnarray}
Furthermore, by multiplying out the left-hand side of
equation~\eqref{eq:r}, we get
\begin{eqnarray} \label{eqn:goal3}
 I-f(x) -f(\ell_r) +f(\ell_r)f(x) & = & 0.
 \end{eqnarray}
 Since the variable $x$ and the literal $\ell_r$ appear in the same
 clause of the instance $\mathcal{I}$, namely, the clause
 $(\overline{\ell_r} \lor \neg x)$, we have that $f(x)f(\ell_r) =
 f(\ell_r)f(x)$. Therefore, by adding equations~\eqref{eqn:goal2}
 and~\eqref{eqn:goal3}, we get that $2I-2f(x)=0$, which implies that
 $f(x)= I$.

  An entirely symmetric argument using the path from $\neg x$ to $x$,
  instead of the path from $x$ to $\neg x$, gives $f(x) = -I$, which
  contradicts the previous finding that $f(x)=I$.
\end{proof}

\begin{proof}[Proof of Lemma~\ref{lem:twosatandhornsat} for HORN SAT]
Let $\mathcal{I}$ be a Horn formula.  As with the proof for  2SAT, the
only non-trivial direction is $\emph{3} \Longrightarrow \emph{1}$.  To
prove the implication $\emph{3} \Longrightarrow \emph{1}$, assume that
$f$ is a satisfying operator assigment for $\mathcal{I}$ over a
(finite-dimensional or infinite-dimensional) Hilbert space $\mathcal{H}$, and,
towards a contradiction, assume that $\mathcal{I}$ is
unsatisfiable in the Boolean domain. As in the proof for 2SAT, let $f$ be extended to all
literals by $f(\ell) = \sg(\ell) f(x)$, where $x$ is the variable
underlying $x$. We will make use of the characterization of
 unsatisfiable in the Boolean domain Horn instances in terms of unit resolution.
For this, we need to first introduce some terminology and notation. If
$C$ and $C'$ are two clauses such that $C$ contains a literal $\ell$
and $C'$ contains the complementary literal $\overline{\ell}$ of
$\ell$, then the \emph{resolution rule} produces in one step the
\emph{resolvent} clause $D$ that is the disjunction of all literals in
the \emph{premises} $C$ and $C'$ other than $\ell$ and
$\overline{\ell}$. The \emph{unit resolution rule} is the special case
of the resolution rule in which (at least) one of the clauses $C$ and
$C'$ is a single literal.  It is well known (see, e.g., \cite{Schoning2008}) that a Horn formula
$\mathcal{I}$ is unsatisfiable if and only if there is a \emph{unit
  resolution derivation} of the empty clause from the clauses of
$\mathcal{I}$, i.e., there is a sequence $C_1,\ldots,C_m$ of clauses
such that, for each $i \in \{1,\ldots,m\}$, we have that $C_i$ is one
of the clauses of $\mathcal{I}$ or $C_i$ is obtained from earlier
clauses $C_j$ and $C_k$ in the sequence via the unit resolution
rule. Clearly, in a unit resolution derivation of the empty clause,
the last application of the unit resolution rule involves two clauses
each of which is the complementary literal of the other.

In what follows, we will show that a unit resolution derivation
can be ``simulated" by a sequence of equations involving operator
assignments. We begin by formulating and proving the following
claim.

\smallskip

%\begin{claim}
\noindent{\bf Claim 1:}
Let $(\ell_1\lor \cdots \lor \ell_r)$ be
clause and let $\overline{\ell_j}$ be the complementary literal of
some literal $\ell_j$ in that clause. If $f$ satisfies both the clause
$(\ell_1\lor \cdots \lor \ell_r)$ and the literal $\overline{\ell_j}$,
then $f$ also satisfies the resolvent $(\ell_1 \lor \cdots \lor
\ell_{j-1} \lor \ell_{j+1} \lor \cdots \lor \ell_r)$ of $(\ell_1\lor
\cdots \lor \ell_r)$ and $\overline{\ell_j}$; equivalently, the operators
$\{ f(\ell_i) : i \not= j \}$ pairwise commute and
\begin{eqnarray}\label{eq:horn1}
\prod_{i=1}^{j-1} (I+f(\ell_i)) \prod_{i=j+1}^{r} (I+f(\ell_i)) & = & 0.
\end{eqnarray}
%\end{claim}

Observe that for the unit resolution rule, as is the case here, the
resolvent is always a subclause of one of the premises. In particular,
since $f$ satisfies both premises, all the operators involved in the
premises commute, and so do the ones involved in the resolvent clause.
To complete the proof of the claim observe that, since $f$ satisfies
both the clause $(\ell_1\lor \cdots \lor\ell_r)$ and the literal
$\overline{\ell_j}$, the corresponding operators commute, and the
identity of polynomials in commuting variables of
Lemma~\ref{lem:fourierofaclause} implies that
\begin{eqnarray}
\prod_{i=1}^r(I+f(\ell_i)) & = & 0 \label{eq:horn2}\\
(I-f(\ell_j))  & = & 0.  \label{eq:horn3}
\end{eqnarray}
%Since the literals $\ell_1,\ldots,\ell_r$ appear in the same clause, we have that $f(\ell_k)f(\ell_m) = f(\ell_m)f(\ell_k)$, for all $k, m \leq r$.  It follows that
%\begin{eqnarray}
%\prod_{i}(I+\sg(\ell_i)f(\ell_i)) & = & (I+\sg(\ell_j)f(\ell_j))\prod_{i\not = j}(I+\sg(\ell_i)f(\ell_i))
%\end{eqnarray}
%hence,
%\begin{eqnarray}
%(I+\sg(\ell_j)f(\ell_j))\prod_{i\not = j}(I+\sg(\ell_i)f(\ell_i))& = & 0. \label{eq:horn3}
%\end{eqnarray}
%By multiplying Equation (\ref{eq:horn2}) from the right by $\prod_{i\not = j}(I+\sg(\ell_i)f(\ell_i))$, we get
%\begin{eqnarray}
%(I-\sg(\ell_j)f(\ell_j))\prod_{i\not = j}(I+\sg(\ell_i)f(\ell_i))& = & 0. \label{eq:horn4}
%\end{eqnarray}
By multiplying equation~\eqref{eq:horn3} by
$\prod_{i=1}^{j-1}(I+f(\ell_i))$ from the left, and by
$\prod_{i=j+1}^r (I+f(\ell_i))$ from the right, we get
\begin{eqnarray}
\left({\prod_{i=1}^{j-1}(I+f(\ell_i)})\right)(I-f(\ell_j))\left({\prod_{i=j+1}^r(I+f(\ell_i))}\right) & = & 0. \label{eq:horn4}
\end{eqnarray}
By adding equations~\eqref{eq:horn2} and~\eqref{eq:horn4}, we
get~\eqref{eq:horn1}, which completes the proof of Claim 1.

Consider now a unit resolution derivation $C_1,\ldots,C_m$ of the
empty clause from the clauses of $\mathcal{I}$. Since the operator
assignment $f$ satisfies every clause of $\mathcal{I}$, we can apply
Claim~1 repeatedly and, by induction, show that $f$ satisfies each
clause in this derivation. Since $C_m$ is the empty clause, it must
have been derived via the unit resolution rule from two earlier
clauses each of which is the complementary literal of the other, say,
$\ell$ and $\overline{\ell}$. So, we must have $f(\ell)= -I$ and
$f(\overline{\ell})= -I$, which is a contradiction since $f(\ell) =
-f(\overline{\ell})$.
\end{proof}

In what follows, we will use Lemma~\ref{lem:twosatandhornsat} to show that if $A$ is bijunctive or Horn or dual Horn, then $A$ has no gaps of any kind.

Assume that $A$ is bijunctive. Note
that we cannot apply Lemma~\ref{lem:twosatandhornsat} directly to conclude that $A$ has no gaps of any kind, because the relations in the
constraint-language $A$ are defined by conjunctions of 2-clauses, but need not be defined by individual 2-clauses.
 In order to be able to apply Lemma~\ref{lem:twosatandhornsat}, we first need to verify the following claim. Assume that $(Z, R)$ is a constraint in
which $R$ is a relation in $A$ defined by a conjunction $C_1 \wedge
\cdots \wedge C_m$, where each $C_i$ is a 2-clause on the variables in
$Z$. Then a satisfying operator assignment for the instance consisting of the single constraint
 $(Z,R)$ will also satisfy each of the 2-clause constraints
$(W_1,C_1),\ldots,(W_r,C_r)$ individually, where $W_i =
(Z_{c_i},Z_{d_i})$ is the tuple of components of $Z =
(Z_1,\ldots,Z_r)$ that appear in $C_i$. To prove this claim, first note that
the commutativity condition on the operators assigned to the variables
in $W_i$ is guaranteed by the commutativity condition on the variables
in $Z$. Thus, we just need to check that the characteristic polynomial
of $C_i$ evaluates to $-I$, and to do so we use Lemma~\ref{lem:entail}
for an appropriately defined system of polynomial equations. In
the remaining, fix $i \in [m]$.

Our polynomials  have variables $X_1,\ldots,X_r$. Let $Q_1$ be the
polynomial $P_R(X_1,\ldots,X_r) + 1$, so that the equation $Q_1 = 0$
ensures $P_R(X_1,\ldots,X_r) = -1$, where $P_R$ is the characteristic
polynomial of $R$. Let $Q$ be the polynomial $P_{C_i}(X_{c_i},X_{d_i})
+ 1$, so that the equation $Q = 0$ ensures $P_{C_i}(X_{c_i},X_{d_i}) =
-1$, where $P_{C_i}$ is the characteristic polynomial of $C_i$, and
$c_i$ and $d_i$ are the indices of the components of $Z$ in
$W_i$. Then, every  Boolean assignment that satisfies the
equation $Q_1 = 0$ belongs to $R$, from which it follows that the Boolean assignment
satisfies the conjunct $C_i$ in the bijunctive definition of $R$, and
hence it also satisfies the equation $Q = 0$. Thus,
Lemma~\ref{lem:entail} applies and every operator assignment that
satisfies $P_R(Z) = -I$ also satisfies $P_{C_i}(W) = -I$, as was to be
proved.

We are now ready to complete the proof that if $A$ is bijunctive, then $A$ has no gaps. Let $\mathcal{I}$ be an instance over $A$ that is
   satisfiable via operators. The preceding
  paragraph shows that the 2SAT instance that results from replacing
  each constraint in the instance $\mathcal{I}$ by its defining
  conjunction of 2-clauses is also satisfiable via operators.
  By Lemma~\ref{lem:twosatandhornsat}, this  2SAT
  instance is also  satisfiable in the Boolean domain. But then $\mathcal{I}$
  itself is satisfiable in the Boolean domain, as was to be shown.

   If $A$ is Horn or
  dual Horn, then  the proof is entirely analogous.

\subsection{Background on Post's Lattice} \label{subsec:Post}

Before we start with the second part in the proof of
Theorem~\ref{thm:gaps}, we need to introduce some basic terminology
and basic results from universal algebra; we devote this
section to that.

\newcommand{\pol}{\mathrm{Pol}}
\newcommand{\inv}{\mathrm{Inv}}

Let $R \subseteq \{ \pm 1 \}^r$ be a Boolean relation of arity $r$ and
let $f : \{ \pm 1 \}^m \rightarrow \{ \pm 1 \}$ be a Boolean operation
of arity $m$. The relation $R$ is \emph{invariant} under $f$ if, for
all sequences of $m$ many $r$-tuples
$(a_{1,1},\ldots,a_{1,r}),\ldots,(a_{m,1},\ldots,a_{m,r})$ in $\{ \pm
1 \}^r$, the following holds:
\begin{equation}
\begin{array}{lll}
\text{if }\; (a_{1,1},\ldots,a_{1,r}),\ldots,(a_{m,1},\ldots,a_{m,r})\;
\text{ are tuples in } R, \\
\text{then }\; (f(a_{1,1},\ldots,a_{m,1}),\ldots,f(a_{1,r},\ldots,a_{m,r}))\;
\text{ is also a tuple in } R.
\end{array}
\label{eqn:tuple}
\end{equation} 
Note that the tuple in the second line is obtained by applying the
$m$-ary operation $f$ to the $m$ many tuples in the first line
\emph{componentwise}. If $A$ is a Boolean constraint language, we say
that $A$ is invariant under $f$ if every relation in $A$ is invariant
under $f$. Whenever $A$ is invariant under $f$ we also say that $f$ is
a \emph{closure operation} of $A$.

The importance of the closure operations of a constraint language
stems from the fact that they completely determine the relations that
are pp-definable from it. This semantic characterization of the
syntactic notion of pp-definability was discovered by Geiger
\cite{Geiger1968} and, independently, Bodnarchuk et
al.\ \cite{Bodnarchuk1969}, for all constraint languages of arbitrary
but finite domain. Here we state the special case of this
characterization for the Boolean domain, since only this special case
is needed in our applications.

\begin{theorem}[\cite{Geiger1968,Bodnarchuk1969}] \label{thm:geigeretal} 
Let $A$ be a Boolean constraint language and
let $R$ be a Boolean relation. The following statements are equivalent:
\begin{enumerate} \itemsep=0pt
\item $R$ is pp-definable from $A$ by a pp-formula without constants,
\item $R$ is invariant under all Boolean closure operations of $A$.
\end{enumerate}
\end{theorem}

In the following we refer to Theorem~\ref{thm:geigeretal} as Geiger's
Theorem.

Recall from Section~\ref{sec:pp} that a pp-formula without constants
is one in which the constants $+1$ and $-1$ do not appear in its
quantifier-free part of the formula. Although it will not be used
until a later section, it is worth pointing out here that a similar
characterization of pp-definability \emph{with} constants
exists. Indeed, it is easy to see that Geiger's Theorem implies that a
Boolean relation $R$ is pp-definable from the Boolean constraint
language $A$ by a pp-formula with constants if and only if $R$ is
invariant under all \emph{idempotent} Boolean closure operations
of~$A$, or equivalently, invariant under all Boolean closure
operations of the Boolean constraint language $A^+$ that is obtained
from $A$ by adding the two unary \emph{singleton} relations $\{+1\}$
and $\{-1\}$; i.e., $A^+ = A \cup \{\{+1\},\{-1\}\}$. We return to the
issue of definability with constants
in~Section~\ref{sec:closureoperations}.

For every set $F$ of Boolean operations, let $\inv(F)$ denote the set
of all Boolean relations that are invariant under all operations in
$F$.  Conversely, for every set of Boolean relations~$A$, let
$\pol(A)$ denote the set of all Boolean operations under which all
relations in $A$ are invariant.  Geiger's Theorem implies that the
mappings $A \mapsto \pol(A)$ and $F \mapsto \inv(F)$ are the lower and
upper adjoints of a Galois connection \cite{DaveyPriestley2002}
between the partially ordered set of sets of Boolean relations ordered
by inclusion, and the partially ordered set of sets of Boolean
operations, also ordered by inclusion.

Note that for every constraint language $A$, the set $\pol(A)$
contains all \emph{projection operations}: all operations $f : \{ \pm
1 \}^r \rightarrow \{ \pm 1 \}$ for which there exists an index $i \in
[r]$ such that $f(x_1,\ldots,x_r) = x_i$ for all $(x_1,\ldots,x_r) \in
\{ \pm 1 \}^r$. Also, $\pol(A)$ is closed under \emph{compositions}:
if $f : \{ \pm 1 \}^s \rightarrow \{ \pm 1 \}$ and $g_1,\ldots,g_s :
\{ \pm 1 \}^r \rightarrow \{ \pm 1 \}$ are operations from $\pol(A)$,
then the operation $h = f \circ (g_1,\ldots,g_s)$ defined by
$h(x_1,\ldots,x_r) =
f(g_1(x_1,\ldots,x_r),\ldots,g_s(x_1,\ldots,x_r))$ for all
$(x_1,\ldots,x_r) \in \{ \pm 1 \}^r$ is also in $\pol(A)$. Any set of
relations that contains all projection operations and that is closed
under compositions is called a \emph{clone}.

Post \cite{Post1941} analyzed the collection of all clones of Boolean
operations and completely determined the inclusions between them. In
particular, he showed that this collection forms a lattice under
inclusion, which is known as Post's lattice. 
%
% We write $\inv(\pol(A))$ to denote the set of all Boolean relations
% such that are closed under every function in $\pol(A)$.  Geiger
% \cite{Geiger1968} and, independently, Bodnarchuk et al.\
% \cite{Bodnarchuk1969} showed that $\inv(\pol(A))$ is the set of all
% Boolean relations $P$ that are pp-definable from $A$ without using
% constants, i.e., the primitive-positive formula that defines $P$
% from relations in $A$ contains no constants.  In what follows, we
% will refer to this result as Geiger's Theorem.
%
In denoting clones in Post's lattice, we will follow the notation and
terminology used by B\"ohler et al.\ \cite{Bohler2003}. The lattice is
represented by the diagram in Figure~\ref{fig:postslattice}, which is
also borrowed from~\cite{Bohler2003} (we thank Steffen
Reith for allowing us to reproduce the diagram here).

\begin{figure}
\begin{center}
\includegraphics[scale=0.75]{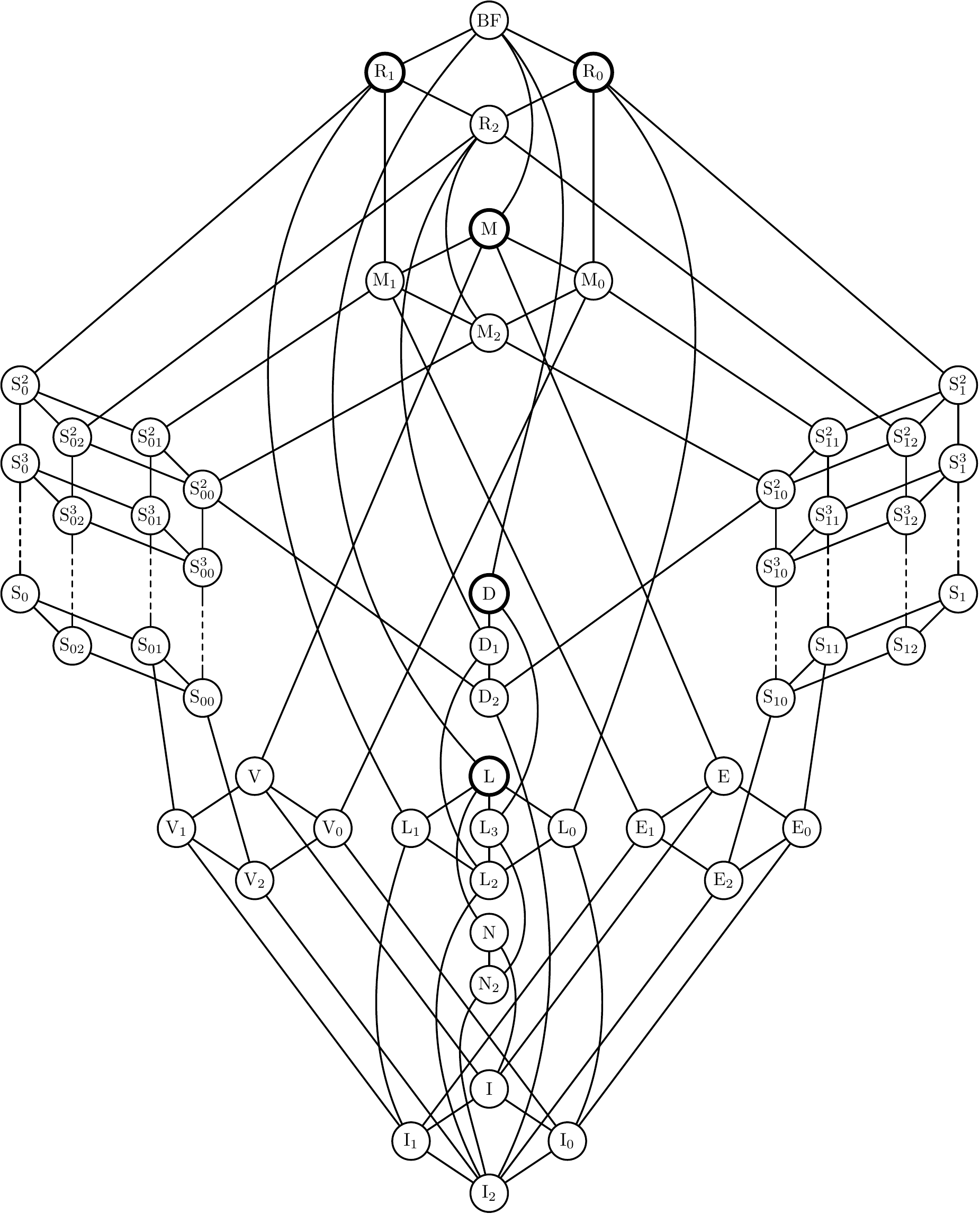}
\end{center}
\caption{Graph of all Boolean clones (diagram by Steffen Reith).}
\label{fig:postslattice}
\end{figure}

Each circle in the diagram of Figure~\ref{fig:postslattice} represents
a clone of Boolean operations, and a line between two circles denotes
inclusion of the clone of the lower circle into the clone of the upper
circle. Post showed that every clone of Boolean operations is
represented in the diagram. Post also identified a \emph{finite basis}
of operations for each clone, which means that the clone is the
smallest class of operations that contains the operations in the basis
and all the projections operations, and that is closed under
composition. For our application, we need only the bases for the eight
clones called $\mathrm{I}_2$, $\mathrm{I}_0$, $\mathrm{I}_1$,
$\mathrm{D}_2$, $\mathrm{E}_2$, $\mathrm{V}_2$, $\mathrm{L}_2$ and
$\mathrm{N}_2$. These are listed in the table in
Figure~\ref{fig:bases}.

\begin{figure}
\begin{center}
\begin{tabular}{l|lll|l}
\cline{1-2} \cline{4-5}
$\mathrm{I}_2$ & $\emptyset$ & \;\;\; & $\mathrm{E}_2$ & $\{x \wedge y\}$ \\
\cline{1-2} \cline{4-5}
$\mathrm{I}_0$ & $\{\text{false}\}$ & \;\;\; & $\mathrm{V}_2$ & $\{x \vee y\}$ \\
\cline{1-2} \cline{4-5}
$\mathrm{I}_1$ & $\{\text{true}\}$ & \;\;\; & $\mathrm{L}_2$ & $\{x \oplus y \oplus z\}$ \\
\cline{1-2} \cline{4-5}
$\mathrm{D}_2$ & $\{ (x \wedge y) \vee (x \wedge z) \vee (y \wedge z) \}$ & \;\;\; & $\mathrm{N}_2$ & $\{\neg x\}$ \\
\cline{1-2} \cline{4-5}
\end{tabular}
\end{center}
\caption{Bases of some selected clones from
  Figure~\ref{fig:postslattice}. Here $\wedge$, $\vee$, $\neg$ and
  $\oplus$ denote Boolean conjunction, Boolean disjunction, Boolean
  negation, and Boolean exclusive or, respectively.}
\label{fig:bases}
\end{figure}

The final ingredient we need from Post's lattice is a characterization
of the tractable Boolean constraint languages from Schaefer's Theorem
in terms of their closure operations.

\begin{theorem}[see Section 1.1 in \cite{Bohler2004}]
\label{thm:conditions}
Let $A$ be a Boolean constraint language. The following statements
hold.
\begin{enumerate} \itemsep=0pt
\item $A$ is $0$-valid if and only if $A$ is invariant under
the constant $\mathrm{false}$ operation.
\item $A$ is $1$-valid if and only if $A$ is invariant under the
  constant $\mathrm{true}$ operation.
\item $A$ is bijunctive if and only if $A$ is invariant under
$(x \wedge y) \vee (x \wedge z) \vee (y \wedge z)$.
\item $A$ is Horn if and only if $A$ is invariant under $x \wedge y$.
\item $A$ is dual Horn if and only if $A$ is invariant under $x \vee y$.
\item $A$ is affine if and only if $A$ is invariant under $x \oplus y
\oplus z$.
\end{enumerate} 
\end{theorem}

\noindent For the connection with Post's lattice, note that, by
Figure~\ref{fig:bases}, the six conditions listed on the right of the
entries \emph{1} through \emph{6} in Theorem~\ref{thm:conditions} are
equivalent to $\pol(A)$ containing the clones $\mathrm{I}_0$,
$\mathrm{I}_1$, $\mathrm{D}_2$, $\mathrm{E}_2$, $\mathrm{V}_2$ and
$\mathrm{E}_2$, respectively.

\subsection{Gaps of Every Kind} \label{subsec:all-gaps}

We are ready to proceed with the second part in the proof of
Theorem~\ref{thm:gaps}. Assume that $A$ satisfies none of the
conditions in the fourth statement in Theorem~\ref{thm:gaps}, i.e.,
$A$ is not $0$-valid, $A$ is not $1$-valid, $A$ is not bijunctive, $A$
is not Horn, and $A$ is not dual Horn. We will show that $A$ has a
satisfiability gap of the first kind (hence, $A$ also has a
satisfiability gap of the second kind) and $A \cup \{{\rm T}\}$ has a
satisfiabiilty gap of the third kind.

As a stepping stone, we will use the known fact that LIN has gaps of every kind. We now  discuss the proof of this fact  and give the appropriate references to the literature.

Recall that LIN is the class of all affine relations, i.e., Boolean relations that are the set of solutions of a system of linear equations over the two-element field. In the $\pm 1$-representation, every such equation is a \emph{parity} equation of the form $\prod_{i=1}^r x_i = y$, where $y \in \{\pm 1\}$.

Mermin \cite{Mermin1990,Mermin1993} considered the following system $\mathcal M$ of parity equations:
\begin{equation}
\begin{array}{ccccccc}
 X_{11} X_{12} X_{13} & = & 1 & \;\;\;\; & X_{11} X_{21} X_{31} & = & 1 \\
 X_{21} X_{22} X_{23} & = & 1 & &  X_{12} X_{22} X_{32} & = & 1 \\
 X_{31} X_{32} X_{33} & = & 1 & &  X_{13} X_{23} X_{33} & = & -1. \\
\end{array}
\label{eqn:Merminequations}
\end{equation}
Graphically, this system of equations can be represented by a square,
where each equation on the left of~\eqref{eqn:Merminequations} comes
from a row, and each equation on the right
of~\eqref{eqn:Merminequations} comes from a column.
\begin{center}
\begin{tikzpicture}[draw=black, ultra thick, x=\Size,y=\Size]
\node[Square] at (1,4) { $X_{11}$ };
\node[Square] at (2,4) { $X_{12}$ };
\node[Square] at (3,4) { $X_{13}$ };
\node[Square] at (1,3) { $X_{21}$ };
\node[Square] at (2,3) { $X_{22}$ };
\node[Square] at (3,3) { $X_{23}$ };
\node[Square] at (1,2) { $X_{31}$ };
\node[Square] at (2,2) { $X_{32}$ };
\node[Square] at (3,2) { $X_{33}$ };
\node at (4,4) { $+1$ };
\node at (4,3) { $+1$ };
\node at (4,2) { $+1$ };
\node at (1,1) { $+1$ };
\node at (2,1) { $+1$ };
\node at (3,1) { $-1$ };
\end{tikzpicture}
\end{center}
%
% \begin{equation}
% \begin{array}{|c|c|c|c}
% \cline{1-3}
% & & & \\
% \;\;X_{11}\;\; & \;\;X_{12}\;\; & \;\;X_{13}\;\; & +1 \\
% & & & \\
% \cline{1-3}
% & & & \\
% \;\;X_{21}\;\; & \;\;X_{22}\;\; & \;\;X_{32}\;\; & +1 \\
% & & & \\
% \cline{1-3}
% & & & \\
% \;\;X_{31}\;\; & \;\;X_{32}\;\; & \;\;X_{33}\;\; & +1 \\
% & & & \\
% \cline{1-3}
% \multicolumn{1}{l}{\;\;+1\;\;} & \multicolumn{1}{l}{\;\;+1\;\;} & 
% \multicolumn{1}{l}{\;\;-1\;\;} & \multicolumn{1}{l}{}
% \end{array}
% \end{equation}
It is easy to see that this system of equations has no solutions in
the Boolean domain. Indeed, by multiplying the left-hand sides of all
equations, we get $1$ because every variable $X_{ij}$ occurs 
twice in the system and $X_{ij}^2 = 1$. At the same time, by
multiplying the right-hand sides of all equations, we get $-1$, hence
the system has no solutions in the Boolean domain. Observe, however,
that this argument used the assumption that variables commute
pairwise, even if they do not appear in the same equation. Thus, this
argument does not go through if one assumes only that variables
occurring in the same equation commute pairwise.  Mermin
\cite{Mermin1990,Mermin1993} showed that the system $\mathcal M$ has a
solution consisting of linear operators on a Hilbert space of
dimension four.  Thus, in our terminology, Mermin established the
following result.

\begin{theorem}[\cite{Mermin1990,Mermin1993}] \label{lem:Mermin}
 $\mathcal{M}$ witnesses a  satisfiability gap of the first kind for {\rm LIN}.
\end{theorem}

Cleve and Mittal \cite[Theorem 1]{CleveM14}  have shown that a system of parity equations has a solution consisting of linear operators on a finite-dimensional Hilbert space if and only if there is a perfect strategy in a certain non-local game in the tensor-product model.  Cleve, Liu, and Slofstra \cite[Theorem 4]{CleveLS16} have shown that
a system of parity equations has a solution consisting of linear operators on a (finite-dimensional or infinite-dimensional) Hilbert space if and only if there is a perfect strategy in a certain non-local game in the commuting-operator model. Slofstra \cite{Slofstra2016} obtained a breakthrough result that has numerous consequences about these models. In particular,  Corollary 3.2 in Slofstra's paper \cite{Slofstra2016} asserts that there is a system $\mathcal S$ of parity equations whose associated non-local game has a perfect strategy in the commuting-operator model, but not in the tensor-product model. Thus, by combining Theorem 1 in \cite{CleveM14}, Theorem 4 in \cite{CleveLS16}, and Corollary 3.2 in \cite{Slofstra2016}, we obtain the following result.

\begin{theorem}[\cite{CleveLS16,CleveM14,Slofstra2016}] \label{thm:Slofstragap}
 $\mathcal{S}$ witnesses a  satisfiability gap of
the third kind for ${\rm LIN}$.
\end{theorem}

LIN has a rather special place among all classes of Boolean relations
that are not $0$-valid, are not $1$-valid, are not bijunctive, are not
Horn, and are not dual Horn. This special role is captured by the next
lemma, which follows from Post's analysis of the lattice of clones of
Boolean functions from Section~\ref{subsec:Post}.

\begin{lemma} \label{lem:Post}
Let $A$ be a Boolean constraint language. If $A$ is not $0$-valid,
not $1$-valid, not bijunctive, not Horn, and not dual Horn, then
${\rm LIN}$ is pp-definable from $A$.
\end{lemma}
\begin{proof}
Assume that $A$ is a Boolean constraint language satisfying the
hypothesis of Lemma~\ref{lem:Post}.  We consider the clone $\pol(A)$
and distinguish several cases using Post's lattice. 
% a diagram of
% which is given by Figure 2 in page 7 of \cite{Bohler2003}. We will
% also use facts about the functions generating the clones in Post's
% lattice that can be found in \cite{Bohler2003}, as well as criteria
% about a relation being bijunctive, Horn, or dual Horn that can be
% found in \cite{Bohler2004}.

If $\pol(A)$ is the smallest clone ${\mathrm I}_2$ in Post's lattice,
then $\pol(A)$ contains only the projection functions; hence, every
Boolean relation is closed under every function in $\pol(A)$. Geiger's
Theorem implies that every Boolean relation and, in particular, every
relation in LIN, is pp-definable from $A$ (and, in fact, it is
pp-definable without using constants).
 
 If $\pol(A)$ is not the smallest clone ${\mathrm I}_2$ in Post's lattice, then it must contain one of the seven minimal clones ${\mathrm I}_0$, $\mathrm{I}_1$, ${\mathrm D}_2$, ${\mathrm E}_2$, ${\mathrm V}_2$, ${\mathrm L}_2$, ${\mathrm N}_2$ that contain ${\mathrm I}_2$. Recall that these clones have bases of operations as described in Figure~\ref{fig:bases}. Since $A$ is not $i$-valid, where $i=0,1$, and since the clone ${\mathrm I}_i$ is generated by the constant function $c_i(x)=i$,   it must be the case that $\pol(A)$ does not contain  the clone ${\mathrm I}_0$ or the clone $\mathrm{I}_1$.  Since $A$ is not bijunctive, there is a relation  in $A$ that is not closed under the \emph{majority} function $\mathrm{maj}(x,y,z) = (x\land y) \lor (x\land z) \lor (y\land z)$.  Since the clone ${\mathrm D}_2$ is generated by the function $\mathrm{maj}(x,y,z)$, it must be the case that that $\pol(A)$ does not contain  the clone ${\mathrm D}_2$.  Since $A$ is not Horn,
there is a relation  in $A$ that is not closed under the function $\mathrm{and}(x,y) = x\land y$.  Since the clone ${\mathrm E}_2$ is generated by  the function $\mathrm{and}(x,y)$, it must be the case  that $\pol(A)$ does not contain  the clone ${\mathrm E}_2$. Since $A$ is not dual Horn,
there is a relation  in $A$ that is not closed under the function  $\mathrm{or}(x,y) = x\lor y$.  Since the clone ${\mathrm V}_2$ is generated by the function $\mathrm{or}(x,y)$, it must be the case that that $\pol(A)$ does not contain  the clone ${\mathrm V}_2$.

The preceding analysis shows that there are just two possibilities:
$\pol(A)$ contains the clone ${\mathrm L}_2$ or $\pol(A)$ contains the
clone ${\mathrm N}_2$. Assume first that $\pol(A)$ contains the clone
${\mathrm L}_2$. Since ${\mathrm L}_2$ is generated by the
\emph{exclusive~or} function $\oplus(x,y,z) = x\oplus y\oplus z$ and
since a relation is affine if and only if it is closed under the
function $\oplus$, Geiger's Theorem implies that a relation is
pp-definable without constants from $A$ if and only if it is an affine
relation. Thus, LIN is pp-definable from $A$ (and, in fact, it is
pp-definable without constants).  Finally, assume that $\pol(A)$
contains the clone ${\mathrm N}_2$. Since $\pol(A)$ is generated by
the function $\mathrm{not}(x) = \neg x$, Geiger's Theorem implies that
a relation is pp-definable without constants from $A$ if and only if
it is it is closed under the function $\mathrm{not}(x)$. In
particular, for every $n\geq 1$ and for $i=0,1$, the affine relation
that is the set of solutions of the equation $x_1+ \cdots +x_{2n} = i
\text{ mod } 2$ is pp-definable without constants from $A$. By using
the constant $0$ in these equations, we have that for every $n\geq 1$
and for every $i=0,1$, the affine relation that is the set of
solutions of the equation $x_1+\cdots +x_{2n-1} = i \text{ mod } 2$ is
pp-definable from $A$ (recall that pp-definitions allow constants). It
follows that LIN is pp-definable from $A$.
\end{proof}

The final lemma in this section asserts that reductions based on pp-definitions preserve satisfiability gaps upwards.

\begin{lemma} \label{lem:gapfirstkindupwards}
Let $B$ and $C$ be  Boolean constraint languages such that~$B$~is~pp-definable~from~$C$.
\begin{enumerate} \itemsep=0pt
\item If  $B$ has a satisfiability gap of the first kind, then so does $C$.
\item If $B$ has a satisfiability gap of the third kind, then so does $C \cup \{{\rm T}\}$.
\end{enumerate}
\end{lemma}

\begin{proof} For the first part,
assume that $B$ is pp-definable from $C$ and that $\mathcal{I}$ is an
instance  that witnesses a satisfiability gap of the first kind for $B$. Thus,
$\mathcal{I}$ is satisfiable via fd-operators, but is not
 satisfiable in the Boolean domain. Let $\mathcal{J}$ be the instance over $C$ as
defined in Section~\ref{sec:basicconstruction}. On the one hand, by
Lemma~\ref{lem:liftonly}, the instance $\mathcal{J}$ is also
 satisfiable via fd-operator. On the other hand, by
Lemma~\ref{lem:ppclassical}, the instance $\mathcal{J}$ is also not
 satisfiable in the Boolean domain. Thus, $\mathcal{J}$ witnesses a satisfiability  gap  of the first kind for~$C$.

For the second part, assume that $B$ is pp-definable from $C$ and that $\mathcal{I}$ is an
instance that witnesses a  satisfiability gap of the third kind
for $B$. Thus, $\mathcal{I}$ is satisfiable via operators,  but it is not satisfiable via fd-operators. Let
$\mathcal{\hat{J}}$ be the instance over $C \cup \{ \mathrm{T} \}$ as
defined in Section~\ref{sec:extendedconstruction}. By
Lemma~\ref{lem:liftandproject}, the instance $\mathcal{\hat{J}}$ is
 satisfiable via operators, but it is  not
 satisfiable via fd-operators. Thus, $\mathcal{\hat{J}}$
witnesses a satisfiability  gap of the third kind for $C \cup \{ \mathrm{T} \}$.
\end{proof}

We now have all the machinery needed to put everything together.

Let $A$ be a Boolean constraint language that is not $0$-valid, not $1$-valid, not bijunctive, not Horn, and not dual Horn. By Lemma~\ref{lem:Post}, we have that LIN is pp-definable from $A$. Since, by Theorem~\ref{lem:Mermin}, LIN has a satisfiability gap of the first kind, the first part of Lemma~\ref{lem:gapfirstkindupwards} implies that $A$ has a satisfiability gap of the first kind. Since, by Theorem~\ref{thm:Slofstragap}, LIN has a satisfiability gap of the third kind, the second part of Lemma~\ref{lem:gapfirstkindupwards} implies that $A$ has a satisfiability gap of the third kind. The proof of Theorem~\ref{thm:gaps} is now complete.

%% file: section-6-further-applications-long.tex
\section{Further Applications} \label{sec:applications}

In this section we discuss two applications of the results from
Sections~\ref{sec:pp} and~\ref{sec:gaps}. The first application is
about classification theorems in the style of Schaefer. The second
application builds on Slofstra's results to answer some open questions
from \cite{Acin2015} on the quantum realizability of contextuality
scenarios. While these open questions were solved earlier by Fritz
also using Slofstra's results (see \cite{Fritz2016}), our alternative
perspective may still add some value since, as we will see, we obtain
improved, and indeed optimal, parameters.

\subsection{Dichotomy Theorems} \label{sec:dichotomy}

For a Boolean constraint language $A$, let SAT$(A)$ denote
the following decision problem:
\begin{center}
Given an instance $\mathcal{I}$ over
$A$, is $\mathcal{I}$ satisfiable in the Boolean domain?
\end{center}
Similarly, let SAT$^*(A)$ and SAT$^{**}(A)$ be the versions of the
problem in which the questions are whether $\mathcal{I}$ is
satisfiable via an operator assignment on a finite-dimensional Hilbert
space, or on an arbitrary Hilbert space, respectively.  We say that a
problem poly-m-reduces to another if there is a polynomial-time
computable function that transforms instances of the first problem
into instances of the second in such a way that the answer is
preserved.

Recall that $\mathrm{T}$ denotes the full binary Boolean relation $\{
\pm 1 \}^2$. The construction in
Section~\ref{sec:extendedconstruction} and
Lemma~\ref{lem:liftandproject} give the following:

\begin{lemma} \label{lem:reductions}
Let $A$ and $B$ be Boolean constraint languages and let $A' = A \cup
\{ \mathrm{T} \}$ and $B' = B \cup \{ \mathrm{T} \}$. If $A$ is
pp-definable from $B$, then
\begin{enumerate} \itemsep=0pt
\item {\rm SAT}$(A')$ poly-m-reduces to {\rm SAT}$(B')$,
{\rm SAT}$^*(B')$, and {\rm SAT}$^{**}(B')$,
\item {\rm SAT}$^{*}(A')$ poly-m-reduces to {\rm SAT}$^{*}(B')$.
\item {\rm SAT}$^{**}(A')$ poly-m-reduces to {\rm SAT}$^{**}(B')$.
\end{enumerate}
\end{lemma}

Slofstra's Corollary 3.3 in \cite{Slofstra2016} in combination with
Theorem~4 in \cite{CleveLS16} gives the undecidability of
SAT$^{**}({\rm LIN})$ which, from now on we denote by LIN SAT$^{**}$.

\begin{theorem}[\cite{Slofstra2016},\cite{CleveLS16}] \label{thm:undec}
{\rm LIN SAT}$^{**}$ is undecidable.
\end{theorem}

In combination with
Lemmas~\ref{lem:twosatandhornsat},~\ref{lem:reductions},
and~\ref{lem:Post}, we get the following dichotomy theorem:

\begin{theorem} \label{thm:dichotomy}
Let $A$ be a Boolean constraint language and let $A' = A \cup \{
\mathrm{T} \}$. Then, exactly one of the following holds:
\begin{enumerate} \itemsep=0pt
\item {\rm SAT}$^{**}(A')$ is decidable in polynomial time,
\item {\rm SAT}$^{**}(A')$ is undecidable.
\end{enumerate}
Moreover, the first case holds if and only if $A$ is 1-valid, or $A$
is 0-valid, or $A$ is bijunctive, or $A$ is Horn, or $A$ is dual Horn.
\end{theorem}

\begin{proof}
If $A$ is 1-valid, 0-valid, bijunctive, Horn, or dual Horn, then $A'$
is also of the same type; indeed $\mathrm{T}$ is both 1-valid and
0-valid, and it is also bijunctive, Horn and dual Horn since it is
defined by the empty conjunction of any kind of clauses.  Thus
SAT$^{**}(A')$ is the same problem as SAT$(A')$ by
Lemma~\ref{lem:twosatandhornsat}, which is solvable in polynomial
time.

If on the contrary $A$ is neither 1-valid, nor 0-valid, nor
bijunctive, nor Horn, nor dual Horn, then Lemma~\ref{lem:Post} applies
and LIN has a pp-definition from $A$. In such a case
Lemma~\ref{lem:reductions} applies and SAT$^{**}({\rm LIN}')$ reduces
to SAT$^{**}(A')$, where ${\rm LIN}'$ denotes ${\rm LIN} \cup \{
\mathrm{T} \}$. Since every instance of LIN SAT$^{**}$ is also an
instance of SAT$^{**}({\rm LIN}')$, the undecidability of
SAT$^{**}(A')$ follows from Theorem~\ref{thm:undec}.
\end{proof}

Note that, in case \emph{2}, Theorem~\ref{thm:dichotomy} states that
SAT$^{**}(A')$ is undecidable but it says nothing about
SAT$^{**}(A)$. Luckily, in most cases it is possible to infer
the undecidability of SAT$^{**}(A)$ from the undecidability
of SAT$^{**}(A')$. This is the case, for example, for
both
\begin{align*}
& {\rm 3SAT} = \{ \{ \pm 1 \}^3 \setminus \{ (a_1,a_2,a_3)
\} : a_1,a_2,a_3 \in \{ \pm 1 \} \}, \\
& {\rm 3LIN} = \{ \{ (a_1,a_2,a_3) \in \{ \pm 1 \}^3 :
a_1a_2a_3 = b \} : b \in \{ \pm 1 \} \}.
\end{align*}
In the following we write 3LIN SAT$^{*}$ and 3LIN SAT$^{**}$ to denote
the problems SAT$^*(A)$ and SAT$^{**}(A)$ for $A = {\rm
  3LIN}$. Similarly, we use 3SAT$^*$ and 3SAT$^{**}$ to denote
SAT$^*(A)$ and SAT$^{**}(A)$ for $A = {\rm 3SAT}$.

\begin{theorem} \label{thm:threesatstarstar}
{\rm 3LIN SAT}$^{**}$ and {\rm 3SAT}$^{**}$ are
both undecidable.
\end{theorem}

\begin{proof} Let $A$ be the Boolean constraint language of 3LIN or 3SAT.
  It follows from Theorem~\ref{thm:dichotomy} that
  SAT$^{**}(A')$ is undecidable. Now we reduce this problem
  to SAT$^{**}(A)$. Take any instance $\mathcal{I}$ over $A'$
  and replace each constraint of the type $((Z_1,Z_2),\mathrm{T})$ by
  an equation $Z_1Z_2Y = -1$ in the case of 3LIN, and a clause $Z_1
  \vee Z_2 \vee Y$ in the case of 3SAT, where $Y$ is a fresh variable
  not used anywhere else in the instance. Let $\mathcal{J}$ be the
  resulting instance. If $f$ is a satisfying operator assignment for
  $\mathcal{I}$, then we claim that an appropriate extension $g$ of
  $f$ is a satisfying operator assignment for $\mathcal{J}$. For 3LIN,
  set $g(Y) = -f(Z_2)f(Z_1)$. For 3SAT, set $g(Y) = -I$. To see that
  this works, first note that $g(Z_1) = f(Z_1)$ and $g(Z_2) = f(Z_2)$
  commute since they appear together in a constraint of
  $\mathcal{I}$. Thus, in both cases $g(Z_1)$, $g(Z_2)$ and $g(Y)$
  pairwise commute. Moreover, in the 3LIN case the assignment $g(Y) =
  -f(Z_2)f(Z_1)$ is chosen so that the equation $g(Z_1)g(Z_2)g(Y) =
  -I$ is satisfied; to check this, multiply $g(Y) = -f(Z_2)f(Z_1)$ by
  $g(Z_1)g(Z_2)$ from the right and use $g(Z_2)f(Z_2) = f(Z_2)^2 = I$
  and $g(Z_1)f(Z_1) = f(Z_1)^2 = I$. Also, in the 3SAT case the
  assignment $g(Y) = -I$ annihilates the product in the expression of
  the characteristic polynomial of the clause $Z_1 \vee Z_2 \vee Y$ in
  see Lemma~\ref{lem:fourierofaclause}, which makes the characteristic
  polynomial evaluate to $-I$ regardless of what $g(Z_1)$ and $g(Z_2)$
  are. Thus, the new constraints are satisfied by $g$ and the claim is
  proved.

  Conversely, if $g$ is a satisfying operator assignment for
  $\mathcal{J}$, then the restriction of $g$ to the variables of
  $\mathcal{I}$ is a satisfying operator assignment for~$\mathcal{I}$,
  just because the commutativity of $f(Z_1)$ and $f(Z_2)$ is enforced
  by the fact that they appear together in the constraint $Z_1Z_2Y =
  -I$ or $Z_1 \vee Z_2 \vee Y$ of $\mathcal{J}$, and because the
  characteristic polynomial of $\mathrm{T}$ is the constant $-1$.
\end{proof}

The same construction and argument that we used in
Theorem~\ref{thm:threesatstarstar} starting at a gap instance over the
constraint language 3SAT$\,\cup\,\{\mathrm{T}\}$ gives
a gap instance over 3SAT that will be useful later
on.

\begin{corollary} \label{cor:threesatgap} There is an instance over
  the Boolean constraint language {\rm 3SAT} that witnesses a
  satisfiability gap of the third kind; it is satisfiable via operator
  assignments over some Hilbert space but not over a
  finite-dimensional Hilbert space.
\end{corollary}

For the problems SAT$^*(A')$, a trichotomy theorem can be proved: 1)
polynomial-time solvable vs 2) polynomial-time equivalent to
SAT$^*({\rm LIN}')$ vs 3) both SAT$^*({\rm LIN}')$-hard and
NP-hard. Unfortunately, whether SAT$^*({\rm LIN}')$ or SAT$^*({\rm
  LIN})$ are polynomial-time solvable, NP-hard or undecidable is an
open problem.

\subsection{Quantum Realizability of Contextuality Scenarios}

We follow the terminology in the paper by Ac\'{\i}n, Fritz, Leverrier
and Sainz \cite{Acin2015}.  A \emph{contextuality scenario} is a
hypergraph $H$ with set $V(H)$ of vertices and set $E(H) \subseteq
2^{V(H)}$ of edges such that $\bigcup_{e \in E(H)} e = V(H)$. Given a
contextuality scenario $H$, a quantum model for it is, informally, an
assignment of probabilities to the vertices of $H$ that are reproduced
as the observation probabilities of a collection of projective
measurements associated to the edges of $H$, when the measurements are
applied to some quantum state. When a contextuality scenario has at
least one quantum model, one says that $H$ allows quantum models.  As
argued in \cite{Acin2015}, this can be equivalently stated formally,
without any reference to measurements or quantum states, as follows.

We say that a contextuality scenario $H$ \emph{allows a quantum
  model}, or is \emph{quantum realizable}, if there exists a Hilbert
space $\mathcal{H}$ and an assignment of bounded linear operators
$P_v$ on $\mathcal{H}$ to each vertex $v$ in $V(H)$ in such a way
that:
\begin{enumerate} \itemsep=0pt
\item $P_v$ is self-adjoint,
\item $P_v^2 = P_v$ for each $v \in V(H)$,
\item $\sum_{v \in e} P_v = I$ for each $e \in E(H)$.
\end{enumerate}
Note that 1 and 2 together say that each $P_v$ is an orthogonal
projection operator\footnote{Ac\'{\i}n et al. refer to orthogonal
  projection operators as projections, and so we will to avoid
  confusion with the fact that two orthogonal projection operators $P$
  and $Q$ could fail to satisfy $PQ \not= 0$. It may also be worth
  pointing out that linear-algebraic projection \emph{operators} of
  this section are unrelated to the universal-algebraic projection
  \emph{operations} from Section~\ref{subsec:Post}.}, and~3~says that
the projection operators associated to the vertices of each edge
resolve the identity. In~\cite{Acin2015} the question was raised
whether there exist contextuality scenarios that are quantum
realizable but only over infinite-dimensional Hilbert spaces (see
Problem~8.3.2 in \cite{Acin2015}). A related computational question
was also raised: Is it decidable whether a contextuality scenario
given as input allows a quantum state? (see Conjecture 8.3.3 in
\cite{Acin2015}). Following the notation in \cite{Acin2015}, this
problem is called ALLOWS-QUANTUM. The restriction of the problem in
which the input hypergraph has edges of cardinality at most $k$ we
call $k$-ALLOWS-QUANTUM. See~\cite{Acin2015} for a discussion on why
these problems are important, and their relationship to Connes
Embedding Conjecture in functional analysis.

Soon after Slofstra published his results, both questions raised
in~\cite{Acin2015} were answered by Fritz by reduction from Slofstra's
Theorems~\ref{thm:Slofstragap} and~\ref{thm:undec} (see
\cite{Fritz2016}). In particular, Fritz proved that ALLOWS-QUANTUM is
undecidable. In the following we illustrate the methods developped in
the previous sections to give alternative proofs of these results. As
a bonus, our proof also gives optimal parameters; we get hypergraphs
with edges of size at most~3 that separate infinite-dimensional
realizability from finite-dimensional realizability, and we show that
already 3-ALLOWS-QUANTUM is undecidable. In contrast, Fritz' reduction
incurs an exponential loss in the size of the edges of the hypergraphs
with respect to the arity of the constraints in Slofstra's result,
which is a priori not bounded, and the best it can achieve is size~$4$
anyway. Moreover, as we will see, our~$3$ in the maximum size of the
edges is optimal since it turns out that 2-ALLOWS-QUANTUM is decidable
(and even solvable in polynomial time).

Next we show how our methods can be used to answer these
questions. First, notice that there is a clear similarity between the
requirements 1,~2 and~3 in the definition of quantum realization of
$H$ and the requirements that an operator assignment for a collection
of variables $\{X_v : v \in V(H)\}$ associated to the vertices of $H$
must satisfy. For one thing, if we define $A_v = I - 2P_v$ for every
$v \in V(H)$, then each $A_v$ is a bounded self-adjoint linear
operator such that $A_v^2 = I$. Moreover, the fact that the
projections associated to an edge of $H$ resolve the identity implies
that they pairwise commute. Thus, the operators $A_v$ associated to
the vertices of $e$ also pairwise commute for every edge $e$ of
$H$. This means that the assignment $X_v \mapsto A_v$ thus defined is
a valid operator assignment to any instance with constraint scopes
given by the hyperedges of $H$.

However, the condition $\sum_{v \in e} \frac{1}{2}(A_v - I) = I$
implied by condition~3 through the inverse transformation $P_v =
\frac{1}{2}(A_v - I)$ does not correspond directly to a constraint of
the form $P_R(A_v : v \in e) = -I$ for any Boolean relation $R$. This
means that we cannot interpret the quantum realizability problem
directly as an instance of a satisfiability problem via operator
assignments over a Boolean constraint language. However, as it turns
out, the problem that we called 3-ALLOWS-QUANTUM is literally the same
as the arbitrary Hilbert space version 1-IN-3 SAT$^{**}$ of the
problem called 1-IN-3 SAT$^*$ by Ji\footnote{There is an unfortunate
  clash in notation in that the problem 1-IN-3 SAT$^{*}$ studied by
  Ji~\cite{Ji2013} is not the same as the problem that we would call
  SAT$^{*}(\textrm{1-IN-3 SAT})$, where 1-IN-3 SAT is the Boolean
  relation $\{ \{ (-1,+1,+1), (+1,-1,+1), (+1,+1,-1) \} \}$. Note that
  $P_{\textrm{1-IN-3 SAT}}(X_1,X_2,X_3) = \frac{3}{4} X_1X_2X_3 +
  \frac{1}{4} X_1X_2 + \frac{1}{4} X_2X_3 + \frac{1}{4} X_1X_3 -
  \frac{1}{4} X_1 - \frac{1}{4} X_2 - \frac{1}{4} X_3 + \frac{1}{4}$,
  so the difference is that, even though the characteristic polynomial
  equation $P_{\textrm{1-IN-3 SAT}}(X_1,X_2,X_3) = -I$ is satisfied by
  an operator assignment if and only if the resolution of the identity
  equation $\frac{1}{2} (1-X_1) + \frac{1}{2} (1-X_2) + \frac{1}{2}
  (1-X_3) = -I$ is satisfied by the same operator assignment, the two
  polynomials $P_{\textrm{1-IN-3 SAT}}(X_1,X_2,X_3)$ and $\frac{1}{2}
  (1-X_1) + \frac{1}{2} (1-X_2) + \frac{1}{2} (1-X_3)$ are by no means
  the~same.}. Ji proved that 3SAT$^*$ reduces to 1-IN-3 SAT$^*$, and
in view of Theorem~\ref{thm:threesatstarstar}, the question arises
whether 3SAT$^{**}$ also reduces to 1-IN-3 SAT$^{**}$, or to
3-ALLOWS-QUANTUM, which is the same. We show that it does.

Before we can do it, though, we need the following lemma that Ji
proved for finite-dimensional Hilbert spaces (see Lemma~5 in
\cite{Ji2013}), and that we prove for all Hilbert spaces:

\begin{lemma} \label{lem:Jicommutativity}
Let $\mathcal{H}$ a Hilbert space.  For every two projection operators
$P_1$ and $P_2$ of $\mathcal{H}$ that commute, there exist projection
operators $Q_1$, $Q_2$, $Q_3$ and $Q_4$ of $\mathcal{H}$ such that
\begin{align*}
& P_1 + Q_1 + Q_4 = I \\
& P_2 + Q_2 + Q_4 = I \\
& Q_1 + Q_2 + Q_3 = I.
\end{align*}
Conversely, if $P_1,P_2,Q_1,Q_2,Q_3,Q_4$ are projection operators of
$\mathcal{H}$ that satisfy these equations, then $P_1$ and $P_2$
commute.
\end{lemma}

\begin{proof}
To prove the first claim, consider the pp-formula
\begin{equation}
\phi(Z_1,Z_2) =
\exists U_1 \exists U_2 \exists U_3 \exists U_4 (R_{1/3}(Z_1,U_1,U_4)
\wedge R_{1/3}(Z_2,U_2,U_4) \wedge R_{1/3}(U_1,U_2,U_3)),
\end{equation}
where $R_{1/3} = \{(-1,+1,+1),(+1,-1,+1),(+1,+1-1)\}$. It is
straightforward to check that this formula defines the full binary
Boolean relation $\mathrm{T} = \{ \pm 1 \}^2$. Now, let $\mathcal{I}$
be the instance $((Z_1,Z_2),\mathrm{T})$ and let $\mathcal{J}$ be the
instance obtained from $\mathcal{I}$ as in
Section~\ref{sec:basicconstruction}.  Let $f$ be defined by $f(Z_1) =
1 - 2P_1$ and $f(Z_2) = 1 - 2P_2$. Since $P_1$ and $P_2$ commute and
the characteristic polynomial of $\mathrm{T}$ is the constant $-1$,
the assignment $f$ is a satisfying operator assigment for the instance
$((Z_1,Z_2),\mathrm{T})$.  By Lemma~\ref{lem:liftonly}, there exists
$g$ that extends $f$ and is a satisfying operator assignment for
$\mathcal{J}$ over $\mathcal{H}$.  Moreover, $g$ is pairwise commuting
on each block of $\mathcal{J}$. Take $Q_i = (1-g(U_i)/2$ for $i =
1,2,3,4$. Then $Q_1,\ldots,Q_4$ are projection operators, and
$P_1,P_2,Q_1,\ldots,Q_4$ pairwise commute. We claim that they satisfy
the equations in the lemma. To see this we apply
Lemma~\ref{lem:entail}. Since the equation $P_{R_{1/3}}(Z_1,U_1,U_4) =
-1$ entails the equation $(1-Z_1)/2 + (1-U_1)/2 + (1-U_4)/2 = 1$ over
the Boolean domain $\{ \pm 1 \}$, and at the same time $P_1,Q_1,Q_4$
pairwise commute, the equation $P_{R_{1/3}}(g(Z_1),g(U_1),g(U_4)) =
-I$ implies $P_1 + Q_1 + Q_4 = I$ by Lemma~\ref{lem:entail}. For the
other two equations, the argument is the same.

For the converse, we use the following easy to verify identities
discovered via a computer search by Ji (see the proof of Lemma~5 in
\cite{Ji2013}):
\begin{align*}
[P_1 + Q_1 + Q_4 - I, -P_1 + Q_1 + Q_3] & = [P_1,Q_3] + [Q_4,Q_3] \\
[P_2 + Q_2 + Q_4 - I, -P_1] &= [P_1,P_2] + [P_1,Q_2] \\
[Q_1 + Q_2 + Q_3 - I, P_1 + Q_4] &= [Q_2,P_1] + [Q_3,P_1] + [Q_3,Q_4],
\end{align*}
where $[X,Y]$ denotes the commutator polynomial $XY - YX$.  Note that
the equations in the lemma imply that the left-hand sides are all
$0$. On the other hand, using the identity $[X,Y] + [Y,X] = 0$, the
sum of the right-hand sides is $[P_1,P_2]$. This gives $[P_1,P_2] = 0$
and thus $P_1$ and $P_2$ commute.
\end{proof}

\begin{lemma} \label{lem:reductiontorealizability}
{\rm 3SAT}$^{**}$ poly-m-reduces to {\rm 3-ALLOWS-QUANTUM}.
\end{lemma}

\begin{proof}
  Schaefer proved that 3SAT is pp-definable from the constraint
  language given by the single relation $R_{1/3} =
  \{(-1,+1,+1),(+1,-1,+1),(+1,+1,-1)\}$. If in addition to $R_{1/3}$
  we allow also the relations $R_{1/2} = \{(-1,+1),(+1,-1)\}$ and
  $R_{1/1} = \{ -1 \}$, then the pp-definition can be assumed to have
  the property that each atom involves different variables and no
  constants. For example, an atom of the form $R_{1/3}(X,X,Z)$ can be
  replaced by $R_{1/3}(X,Y,Z) \wedge R_{1/2}(X,Y') \wedge
  R_{1/2}(Y',Y)$, where $Y$ and $Y'$ are fresh quantified variables
  that do not appear anywhere else in the formula.

  We use this for the construction in
  Section~\ref{sec:extendedconstruction}.  Let $\mathcal{I}$ be a
  3SAT instance and let $\mathcal{\hat{J}}$ be the instance over the
  Boolean constraint language $A = \{ R_{1/3} , R_{1/2}, R_{1/1},
  \mathrm{T} \}$ given by the construction in
  Section~\ref{sec:extendedconstruction}, using the pp-definition of
  3SAT from $A$.  Starting at $\mathcal{\hat{J}}$ we produce an
  instance of 3-ALLOWS-QUANTUM as follows: Each variable in
  $\mathcal{\hat{J}}$ becomes a vertex in the hypergraph.  Each
  constraint of the type $((Z_1,Z_2,Z_3),R_{1/3})$ becomes a hyperedge
  $\{Z_1,Z_2,Z_3\}$, each constraint of the type $((Z_1,Z_2),R_{1/2})$
  becomes a hyperedge $\{Z_1,Z_2\}$, each constraint of the type
  $(Z,R_{1/1})$ becomes a singleton hyperedge $\{Z\}$, and each
  constraint of the type $((Z_1,Z_2),\mathrm{T})$ introduces four
  fresh vertices $U_1,U_2,U_3,U_4$ and three hyperedges
  $\{Z_1,U_1,U_4\}$, $\{Z_2,U_2,U_4\}$ and $\{U_1,U_2,U_3\}$ in
  correspondance with the equations of Lemma~\ref{lem:Jicommutativity}
  with $Z_1,Z_2$ playing the role of $P_1,P_2$, and $U_1,U_2,U_3,U_4$
  playing the role of $Q_1,Q_2,Q_3,Q_4$. Let $H$ be the hypergraph
  that results from this construction. We claim that for every Hilbert
  space $\mathcal{H}$, the instance $\mathcal{I}$ is satisfiable via
  operator assignments over $\mathcal{H}$ if and only if the
  hypergraph $H$ is quantum realizable over $\mathcal{H}$.

  In the forward direction, let $f$ be a satisfying operator assignment
  for $\mathcal{I}$ over $\mathcal{H}$. By
  Lemma~\ref{lem:liftandproject}, there is a $g$ that extends $f$ and
  is a satisfying operator assignment for $\mathcal{\hat{J}}$ over
  $\mathcal{H}$. Recall now that each vertex of $H$ is indeed a
  variable of $\mathcal{\hat{J}}$, or an additional vertex of the type
  $U_1,U_2,U_3,U_4$ introduced by a constraint of the form
  $((Z_1,Z_2),\mathrm{T})$. For each $v$ of the first type, let $P_v$
  be the projection operator given by $(1-g(v))/2$. For each $v$ of
  the second type, let $P_v$ be the projection given by
  Lemma~\ref{lem:Jicommutativity} for the projection assignment $P_1 =
  P_{Z_1}$ and $P_2 = P_{Z_2}$ with $U_1,U_2,U_3,U_4$ corresponding to
  $Q_1,Q_2,Q_3,Q_4$. Note that $P_1$ and $P_2$ commute, since $Z_1$
  and $Z_2$ appear together in $((Z_1,Z_2),\mathrm{T})$ and hence
  $g(Z_1)$ and $g(Z_2)$ commute, so the lemma applies. We claim that
  this assignment of operators does the job.

We just need to check that the projection operators resolve the
identity on every edge of $H$. For edges of the type $\{Z_1,Z_2,Z_3\}$
introduced by a constraint $((Z_1,Z_2,Z_3),R_{1/3})$ we show this with
an application of Lemma~\ref{lem:entail}: the equation
$P_{R_{1/3}}(Z_1,Z_2,Z_3) = -1$ entails the equation $(1-Z_1)/2 +
(1-Z_2)/2 + (1-Z_3)/2 = 1$ over the Boolean domain $\{ \pm 1 \}$, and
therefore, since $g(Z_1),g(Z_2),g(Z_3)$ pairwise commute, the equation
$P_{R_{1/3}}(g(Z_1),g(Z_2),g(Z_3)) = -I$ implies $P_{Z_1} +
P_{Z_2} + P_{Z_3} = I$ by Lemma~\ref{lem:entail}. For edges of the
types $\{Z_1,Z_2\}$ or $\{Z\}$ introduced by constraints of the types
$((Z_1,Z_2),R_{1/2})$ or $(Z,R_{1/2})$, respectively, the argument is
the same. Finally, for the three edges that come from a constraint of
the form $((Z_1,Z_2),\mathrm{T})$, the claim follows from
Lemma~\ref{lem:Jicommutativity}. This completes one direction of the
reduction.

For the other direction, let $v \mapsto P_v$ be an assignment of
projection operators of $\mathcal{H}$ that witnesses that $H$ is
quantum realizable. Recall again that each vertex $v$ of $H$ is a
variable of $\mathcal{\hat{J}}$, or an additional vertex of the type
$U_1,U_2,U_3,U_4$ coming from a $\mathrm{T}$-constraint. For each $v$
of the first type, let $A_v = I-2P_v$. Each $A_v$ is a self-adjoint
bounded linear operator that squares to the identity. Moreover, any
two variables of $\mathcal{\hat{J}}$ that appear together in a
constraint that is not a $\mathrm{T}$-constraint appear together as
vertices in some edge of $H$. Therefore the corresponding operators
belong to the resolution of the identity of that edge, and a set of
projection operators that resolve the identity are pairwise orthogonal
and hence commute. Also, for any two variables of $\mathcal{\hat{J}}$
that appear together in a constraint of the form
$((Z_1,Z_2),\mathrm{T})$, the corresponding operators commute thanks
to the ``conversely'' clause in Lemma~\ref{lem:Jicommutativity}. Thus,
the only thing left to do is checking that each constraint of
$\mathcal{\hat{J}}$ is satisfied.

For constraints of the type $((Z_1,Z_2,Z_3),R_{1/3})$ this follows
also from an application of Lemma~\ref{lem:entail}: the equation
$(1-Z_1)/2 + (I-Z_2)/2 + (I-Z_3)/2 = 1$ entails the equation
$P_{R_{1/3}}(Z_1,Z_2,Z_3) = -1$ over the Boolean domain $\{ \pm 1 \}$,
and since $A_{Z_1},A_{Z_2},A_{Z_3}$ pairwise commute, the equation
$P_{Z_1} + P_{Z_2} + P_{Z_3} = I$ implies
$P_{R_{1/3}}(A_{Z_1},A_{Z_2},A_{Z_3}) = -I$ by Lemma~\ref{lem:entail}.
For constraints of the type $((Z_1,Z_2),R_{1/2})$ and $(Z,R_{1/1})$
the argument is the same.
\end{proof}

In combination with Theorem~\ref{thm:threesatstarstar} we get
the following.

\begin{corollary} {\rm 3-ALLOWS-QUANTUM} and
{\rm ALLOWS-QUANTUM} are undecidable. \end{corollary}

\noindent The same construction as in
Lemma~\ref{lem:reductiontorealizability} starting from
Corollary~\ref{cor:threesatgap} gives the~next.

\begin{corollary} There exists a hypergraph with edges of size
at most three that is quantum realizable on some Hilbert space, but
not on a finite-dimensional Hilbert space. \end{corollary}

It was mentioned earlier that
2-ALLOWS-QUANTUM is decidable in polynomial
time. One way to see this is by arguing that a hypergraph with edges
of size two (i.e.\ a graph) is quantum realizable if and only if it is
bipartite. Another is by reduction to 2SAT$^{**}$,
which is decidable in polynomial time by
Theorem~\ref{thm:dichotomy}. A close look reveals that, indeed, both
proofs are the same.

\begin{theorem}
{\rm 2-ALLOWS-QUANTUM} is decidable in polynomial time.
\end{theorem}

\begin{proof}
We reduce to 2SAT$^{**}$. Given a hypergraph $H$,
build the 2SAT instance that has one variable $X_v$ for each vertex
in $V(H)$, two clauses $X_u \vee X_v$ and $\neg X_u \vee \neg X_v$ for
every edge $\{u,v\} \in E(H)$, and one unit clause $X_u$ for each
singleton edge $\{u\}$ in $E(H)$. It is straightforward to check that
this reduction works through the usual conversion from projection
operators to involutions $P_v \mapsto 1-2P_v$, and the usual conversion
from involutions to projection operators $A_v \mapsto (1-A_v)/2$.
\end{proof}

%% file: section-7-closure-operations-long.tex
\section{Closure Operations} 
\label{sec:closureoperations}

In this section we develop a generalization of the concept of closure
operation from Section~\ref{subsec:Post} for sets of operator
assignments. 
% For sets of Boolean assignments such closure operations
% are also called polymorphisms, and they play a key role in the theory
% of constraint satisfaction problems. 
For every Boolean $r$-ary
relation $R$, let $R^*$ denote the set of fully commuting $r$-variable
operator assignments over finite-dimensional Hilbert spaces that
satisfy the equation $P_R(X_1,\ldots,X_r) = -I$.  We show that every
closure operation for $R$ gives a suitable closure operation for
$R^*$. As an application, we show that the set of Boolean relations
that are pp-definable from a Boolean constraint language is not
enlarged when we allow the existential quantifiers to range over
operator assignments.

\subsection{Closure Operations and pp$^*$-Definitions}

Let $A$ be a Boolean constraint language and let $R$ be a Boolean
relation of arity~$r$. Let $\psi = R_1(z_1) \wedge \cdots \wedge
R_m(z_m)$ be a conjunction of atoms with relations from $A$;
i.e.\ each~$R_i$ is a relation from~$A$, and each~$z_i$ denotes a
tuple of the appropriate arity made of first-order variables or
constants in $\{ \pm 1 \}$. Each such formula can be thought of as an
instance over~$A$. Concretely, it can be thought of as the instance
$\mathcal{I} = ((Z_1,R_1),\ldots,(Z_m,R_m))$, where each $Z_i$ is
obtained from the corresponding $z_i$ by replacing each first-order
variable~$x$ by a correponding variable~$X$, and leaving all constants
untouched.

Let $\mathcal{H}$ be a finite-dimensional Hilbert space. We say that
$R$ is pp$^*$-definable from~$A$ over~$\mathcal{H}$ if there is a
pp-formula $\phi(x_1,\ldots,x_r) = \exists y_1 \cdots \exists y_s
(\psi(x_1,\ldots,x_r,y_1,\ldots,y_s))$ over $A$, where $\psi$ is a
conjunction as above, such that, for every $a_1,\ldots,a_r \in \{ \pm
1 \}$, the tuple $(a_1,\ldots,a_r)$ is in $R$ if and only if the
instance
\begin{equation}
\psi(x_1/a_1,\ldots,x_r/a_r,y_1/Y_1,\ldots,y_s/Y_s)
\end{equation}
is satisfiable via operator assignments over $\mathcal{H}$. We say
that $R$ is pp$^*$-definable from $A$ if it is pp$^*$-definable from
$A$ over a finite-dimensional Hilbert space. One of the goals of this
section is to prove the following conservativity theorem:

\begin{theorem} \label{thm:collapse} Let $A$ be a Boolean constraint
  language and let $R$ be a Boolean relation. If $R$ is
  pp$^*$-definable from $A$, then $R$ is pp-definable from $A$.
\end{theorem}

In order to prove this we need to develop the concept of closure
operation for sets of operator assignments. Let $r$ be a positive
integer. A \emph{relation of operator assignments} of arity~$r$ is a
set of fully commuting operator assignments for a fixed set of $r$
variables $X_1,\ldots,X_r$. Note that we do not require that all
operator assignments come from the same Hilbert space. The relation is
called \emph{Boolean} if all assignments in it come from a Hilbert
space of dimension~$1$; i.e., from~$\mathbb{C}$.  If $\mathcal{H}$ is
a Hilbert space and $R \subseteq \{\pm 1\}^r$ is a Boolean relation of
arity $r$, we write $R^{\mathcal{H}}$ for the set of fully commuting
operator assignments for $X_1,\ldots,X_r$ over $\mathcal{H}$ that
satisfy the polynomial equation $P_R(X_1,\ldots,X_r) = -I$, where
$P_R$ is the characteristic polynomial of~$R$. We write $R^*$ for the
union of $R^{\mathcal{H}}$ over all finite-dimensional Hilbert
spaces. If $A$ is a set of Boolean relations, define $A^* = \{ R^* : R
\in A \}$.

Let $\mathcal{H}_1,\ldots,\mathcal{H}_m$ and $\mathcal{H}$ be Hilbert
spaces, and let $f$ be a function that takes as inputs~$m$ many linear
operators, one on each $\mathcal{H}_i$, and produces as output a
linear operator on $\mathcal{H}$. We say that $f$ is an
\emph{operation} if the following conditions are satisfied.
\begin{enumerate} \itemsep=0pt
\item[1.] If $A_1,\ldots,A_m$ are $1$-variable operator assignments
  over $\mathcal{H}_1,\ldots,\mathcal{H}_m$, then $f(A_1,\ldots,A_m)$
  is a one-variable operator assignment over $\mathcal{H}$.
\item[2.] If $(A_{1,1},A_{1,2}),\ldots,(A_{m,1},A_{m,2})$ are
  commuting $2$-variable operator assignments over
  $\mathcal{H}_1,\ldots,\mathcal{H}_m$, then
  $(f(A_{1,1},\ldots,A_{m,1}),f(A_{1,2},\ldots,A_{m,2}))$ is a
  commuting two-variable operator assignment over $\mathcal{H}$.
\end{enumerate}
Let $R$ be a relation of operator assignments of arity $r$ and let $F$
be a collection of operations as above. We say that $R$ is
\emph{invariant} under $F$ if for each $f \in F$ the following
additional condition is also satisfied.
\begin{enumerate} \itemsep=0pt
\item[3.] If
  $(A_{1,1},\ldots,A_{1,r}),\ldots,(A_{m,1},\ldots,A_{m,r})$ are fully
  commuting $r$-variable operator assignments over
  $\mathcal{H}_1,\ldots,\mathcal{H}_m$, respectively, and
  $(A_{i,1},\ldots,A_{i,r})$ belongs to $R$ for every $i \in [m]$,
  then $(f(A_{1,1},\ldots,A_{m,1}),\ldots,f(A_{1,r},\ldots,A_{m,r}))$
  is a fully commuting $r$-variable operator assignment over
  $\mathcal{H}$ and belongs to $R$.
\end{enumerate}
If $A$ is a set of relations of operator assignments, we say that $A$
is \emph{invariant} under $F$ if every relation in $A$ is invariant
under $F$. We also say that $F$ is a \emph{closure operation} of~$A$.
A \emph{Boolean closure operation} of $A$ is one in which the
dimensions of all Hilbert spaces involved are $1$; i.e., they are
$\mathbb{C}$. Before we prove the main technical result of this
section, we work out a motivating example.

\subsection{Example: LIN}

In this section we study whether $R^*$ for $R = {\rm LIN}$ has some
closure operation.  In the 0-1-representation of Boolean values, the
function $(X_1,X_2,X_3) \mapsto X_1 \oplus X_2 \oplus X_3$ is a
Boolean closure operation of LIN. In the $\pm 1$-representation of
Boolean values, this is $(X_1,X_2,X_3) \mapsto X_1X_2X_3$. It is
tempting to think that the map $(X_1,X_2,X_3) \mapsto X_1X_2X_3$
applied to linear operators on a Hilbert space could already be a
closure operation for ${\rm LIN}^*$. However, the solution to the
Mermin-Peres magic square equations~\eqref{eqn:Merminequations} is a
counterexample: each row equation is a parity equation with even
right-hand side that is satisfied, but the composition of columns by
the operation $X_1X_2X_3$ gives an operator assignment that satisfies
a parity equation with odd right-hand side.

It turns out that the correct way of generalizing the Boolean
closure operation is not by taking ordinary products, but
Kronecker products. Let $F$ be the function that takes any three
linear operators $X_1,X_2,X_3$ over the same finite-dimensional
Hilbert space and is defined by
\begin{equation}
F(X_1,X_2,X_3) = X_1 \otimes X_2 \otimes X_3.
\end{equation}
Now let $(A_1,\ldots,A_r)$, $(B_1,\ldots,B_r)$ and $(C_1,\ldots,C_r)$
be three fully commuting $r$-variable operator assignments over a
finite-dimensional Hilbert space, say $\mathbb{C}^d$. We think of all
operators as matrices. Take $D_i = F(A_i,B_i,C_i)$ for $i =
1,\ldots,r$. These are Hermitian matrices since the operations of
conjugate transposition and Kronecker product commute. Also
\begin{equation}
D_i D_j
= (A_iA_j) \otimes (B_iB_j) \otimes (C_iC_j)
= (A_jA_i) \otimes (B_jB_i) \otimes (C_jC_i)
= D_j D_i \label{eqn:commute}
\end{equation}
so $D_1,D_2,D_3$ pairwise commute. Equation~\eqref{eqn:commute} also
gives $D_i^2 = (A_i^2) \otimes (B_i^2) \otimes (C_i^2) = I \otimes I
\otimes I = I$, so $(D_1,\ldots,D_r)$ is a fully commuting
$r$-variable operator assignment. Next we consider a relation in LIN,
say $R = \{ (a_1,\ldots,a_r) \in \{ \pm 1 \}^r : a_1 \cdots a_r = b
\}$, with $b \in \{ \pm 1 \}$. Note that its characteristic polynomial
is $P_R(X_1,\ldots,X_r) = -b \cdot X_1 \cdots X_r$.  We show that if
$P_R(A_1,\ldots,A_r) = P_R(B_1,\ldots,B_r) = P_R(C_1,\ldots,C_r) =
-I$, then also $P_R(D_1,\ldots,D_r) = -I$. We have
\begin{equation}
\prod_{i=1}^r D_i = 
\left({\prod_{i=1}^r A_i}\right) \otimes 
\left({\prod_{i=1}^r B_i}\right) \otimes 
\left({\prod_{i=1}^r C_i}\right) =
(bI) \otimes (bI) \otimes (bI) = b^3 I = b I.
\end{equation}
Hence $P_R(D_1,\ldots,D_r) = -b^2 I = -I$.  This shows that $F$ is a
closure operation of ${\rm LIN}^*$.

One consequence of the existence of $F$ as a closure operation of
${\rm LIN}^*$ is that the binary OR relation $\mathrm{OR}_2 = \{ \pm 1
\}^2 \setminus \{ (+1,+1) \}$ is not pp$^*$-definable from LIN.

\begin{theorem} \label{thm:ornotdefinable}
$\mathrm{OR}_2$ is not pp$^*$-definable from $\mathrm{LIN}$.
\end{theorem}

Note that this follows from the more general statement in
Theorem~\ref{thm:collapse} since it is known that the Boolean relation
OR$_2$ is not pp-definable from LIN. Indeed, OR$_2$ is not closed
under the (idempotent) Boolean closure operation $(X_1,X_2,X_3)
\mapsto X_1X_2X_3$ of LIN, since $(-1,-1)$, $(+1,-1)$ and $(-1,+1)$
are all three in the relation OR$_2$, but $(+1,+1)$ is not in
OR$_2$. The undefinability of OR$_2$ from LIN by a pp-formula (with or
without constants) follows from the easy direction in Geiger's
Theorem~\ref{thm:geigeretal}. Since we prove
Theorem~\ref{thm:collapse} below, we omit a proof of
Theorem~\ref{thm:ornotdefinable} at this point.

\subsection{Generalization}

We show that every Boolean closure operation gives a closure operation
for relations of operator assignments over finite-dimensional Hilbert
spaces. In the following, if $X_i$ is a linear operator on a Hilbert
space, $X_i^0$ and $X_i^1$ are to be interpreted as the identity
operator and $X_i$ itself, respectively.  If $S$ is a set, we write
$S(i)$ for the 0-1-indicator of the fact that $i$ is in $S$; i.e.\
$S(i) = 1$ if $i$ is in $S$, and $S(i) = 0$ if $i$ is not in $S$.

\begin{theorem} \label{thm:closure}
  Let $A$ be a Boolean constraint language and let $f : \{\pm 1\}^m
  \rightarrow \{\pm 1\}$ be a Boolean closure operation of $A$. Then
  the function on linear operators on finite-dimensional Hilbert
  spaces defined by
\begin{equation}
F(X_1,\ldots,X_m) = \sum_{S \subseteq [m]} \hat{f}(S) \bigotimes_{i \in [m]}
  X_i^{S(i)}
\end{equation}
is a closure operation of $A^*$.  Moreover, $F(a_1 I,\ldots,a_m I) =
f(a_1,\ldots,a_m) I$ holds for every $(a_1,\ldots,a_m) \in \{ \pm 1
\}^m$.
\end{theorem}

\begin{proof}
  First we show that $F$ is an operation; i.e., it satisfies
  conditions 1 and 2 in the definiton of operation. Let
  $X_1,\ldots,X_m$ be $1$-variable operator assignments over
  $\mathcal{H}_1,\ldots,\mathcal{H}_m$. In particular,
  $X_1,\ldots,X_m$ are all self-adjoint linear operators.  Thus, for
  $S \subseteq [m]$ we have
  \begin{equation}
  \Biggl({\bigotimes_{i
      \in [m]} (X_i)^{S(i)}}\Biggr)^* = \bigotimes_{i \in
    [m]} (X_i^*)^{S(i)} = \bigotimes_{i \in [m]}
  (X_i)^{S(i)}.
  \end{equation}
  From this it follows that $F(X_1,\ldots,X_m)$ is self-adjoint since
  each $\hat{f}(S)$ is a real number.  Next we want to show that
  $F(X_1,\ldots,X_m)^2 = I$. First note that for $S,T \subseteq [m]$,
  their symmetric difference $U = S \Delta T$ and their intersection
  $V = S \cap T$, we have
  \begin{equation}
  \Biggl({\bigotimes_{i \in [m]}
  X_i^{S(i)}}\Biggr)
  \Biggl({\bigotimes_{i \in [m]}
  X_i^{T(i)}}\Biggr) =
  \Biggl({\bigotimes_{i \in [m]}
  X_i^{U(i)}}\Biggr)
  \Biggl({\bigotimes_{i \in [m]}
  (X_i^2)^{V(i)}}\Biggr) =
  \Biggl({\bigotimes_{i \in [m]}
  X_i^{U(i)}}\Biggr),
  \end{equation}
  where the last equality follows from the fact that $X_i^2 = I$ for
  all $i \in [m]$.
%  \begin{equation}
%  \Biggl({\bigotimes_{i \in [m]}
%    (X_i)^{S(i)}}\Biggr)^2 = \bigotimes_{i \in [m]}
%  (X_i^2)^{S(i)} = \bigotimes_{i \in [m]} I = I.
%  \end{equation}
%  From this, $F(X_1,\ldots,X_m)^2 = I$ follows from Parseval's
%  identity
%  \begin{equation}
%  \sum_{S \subseteq [m]} \hat{f}(S)^2 =
%  \frac{1}{2^n} \sum_{x \in \{\pm 1\}^n} f(x)^2 = 1.
%  \end{equation}
  Now we can expand $F(X_1,\dots,X_m)^2$ as follows
  \begin{align}
  F(X_1,\ldots,X_m)^2 & = \sum_{S \subseteq [m]} \sum_{T \subseteq [m]} \hat{f}(S) \hat{f}(T)
  \Biggl({\bigotimes_{i \in [m]}
  X_i^{S(i)}}\Biggr)
  \Biggl({\bigotimes_{i \in [m]}
  X_i^{T(i)}}\Biggr) \\
  & =
  \sum_{S \subseteq [m]} \sum_{U \subseteq [m]} 
  \hat{f}(S) \hat{f}(S \Delta U)
  \Biggl({\bigotimes_{i \in [m]}
  X_i^{U(i)}}\Biggr) \\
  & = 
  \sum_{U \subseteq [m]} \sum_{S \subseteq [m]} 
  \hat{f}(S) \hat{f}(S \Delta U)
  \Biggl({\bigotimes_{i \in [m]}
  X_i^{U(i)}}\Biggr) \\
  & = 
  \sum_{U \subseteq [m]} \Biggl({\bigotimes_{i \in [m]}
  X_i^{U(i)}}\Biggr) \sum_{S \subseteq [m]} 
  \hat{f}(S) \hat{f}(S \Delta U). \label{eqn:lastinseq}
  \end{align}
  By the Convolution Formula~\eqref{eqn:convolution} we have
  \begin{equation}
  \sum_{S \subseteq [m]} \hat{f}(S) \hat{f}(S \Delta U) =
  \widehat{f^2}(U).
  \end{equation}
  Since the range of $f$ is $\{ \pm 1 \}$, the function
  $f^2$ is identically $1$, from which it follows that
  \begin{equation}
  \widehat{f^2}(U) = \left\{ { \begin{array}{lll} 1 & \;\; & \text{ if } U = \emptyset \\ 0 & & \text{ if } U \not= \emptyset\end{array} } \right.
  \end{equation}
  by the uniqueness of the Fourier transform.
  Back into~\eqref{eqn:lastinseq}, this gives
  \begin{equation}
  F(X_1,\ldots,X_m)^2 = \Biggl({\bigotimes_{i \in [m]}
  X_i^{\emptyset(i)}}\Biggr) = I
  \end{equation}
  as was to be proved.  Finally, if $S \subseteq [m]$ and
  $(X_1,Y_1),\ldots,(X_m,Y_m)$ are
  such that $X_i$ and $Y_i$ commute for every $i \in [m]$, then
  \begin{align}
  \Biggl({\bigotimes_{i \in [m]} X_i^{S(i)}}\Biggr)
  \Biggl({\bigotimes_{i \in [m]} Y_i^{S(i)}}\Biggr)
  & = \bigotimes_{i \in [m]} (X_iY_i)^{S(i)} = \\
  & = \bigotimes_{i \in [m]} (Y_iX_i)^{S(i)}
  =
  \Biggl({\bigotimes_{i \in [m]} Y_i^{S(i)}}\Biggr)
  \Biggl({\bigotimes_{i \in [m]} X_i^{S(i)}}\Biggr).
  \end{align}
  It follows that $F(X_1,\ldots,X_m)$ and $F(Y_1,\ldots,Y_m)$ commute.
  This completes the proof that $F$ is an operation.

  Next we show that for every relation $R$ in $A$, the operator
  assignment relation $R^*$ is invariant under $F$. Let $r$ be the
  arity of $R$ and let $P_R(X_1,\ldots,X_r)$ be the characteristic
  polynomial of $R$. Let
  $(A_{1,1},\ldots,A_{1,r}),\ldots,(A_{m,1},\ldots,A_{m,r})$ be
  $r$-variable operator assignments over finite-dimensional Hilbert
  spaces $\mathcal{H}_1,\ldots,\mathcal{H}_m$.  We may assume that
  $\mathcal{H}_i = \mathbb{C}^{d_i}$ where $d_i$ is the dimension of
  $\mathcal{H}_i$. From now on we switch to the language of matrices.

  Assume that all the assignments
  $(A_{1,1},\ldots,A_{1,r}),\ldots,(A_{m,1},\ldots,A_{m,r})$ are in
  $R^*$. In particular, each sequence $A_{i,1},\ldots,A_{i,r}$ is a
  fully commuting assignment of Hermitian matrices and
  $P_R(A_{i,1},\ldots,A_{i,r}) = -I$.  The Strong Spectral Theorem
  (i.e.~Theorem~\ref{thm:sstfd}) applies, so $A_{i,1},\ldots,A_{i,r}$
  simultaneously diagonalize. Let $U_i$ be a unitary matrix of
  $\mathcal{H}_i$ that achieves that, and let $D_{i,j} = U A_{i,j}
  U^*$ for $j \in [r]$ be the resulting diagonal matrices.  From
  $A_{i,j}^2 = I$ and $U^* U = U U^* = I$ we conclude that $D_{i,j}^2
  = I$ and hence each entry in the diagonal of $D_{i,j}$ is $+1$ or
  $-1$. For $c \in [d_i]$, let $D_{i,j}(c)$ denote the entry in
  position $c$ of the diagonal of $D_{i,j}$. The hypotheses of
  Lemma~\ref{lem:simultaneouslysimilar} apply to the pairs
  $(A_{i,1},D_{i,1}),\ldots,(A_{i,r},D_{i,r})$, so
  $P_R(A_{i,1},\ldots,A_{i,r})$ and $P_R(D_{i,j},\ldots,D_{i,r})$ are
  similar matrices. As $P_R(A_{i,1},\ldots,A_{i,r}) = -I$, and the
  only matrix that is similar to $-I$ is $-I$ itself, we get
  $P_R(D_{i,1},\ldots,D_{i,r}) = -I$.  In particular
  \begin{equation}
  P_R(D_{i,1}(c),\ldots,D_{i,r}(c)) = -1 \label{eqn:tuplex}
  \end{equation}
  for every $c \in [d_i]$. This will be of use later.

  Our next goal is to show that
  $P_R(F(A_{1,1},\ldots,A_{m,1}),\ldots,F(A_{1,r},\ldots,A_{m,r})) = -I$
  and we do so by showing that
  \begin{equation}
  \sum_{T \subseteq [r]} \hat{R}(T) \prod_{j \in T}
  F(A_{1,j},\ldots,A_{m,j}) = -I. \label{eqn:goalhere}
  \end{equation}
  For fixed $T \subseteq [r]$, let $A_T = \prod_{j \in T}
  F(A_{1,j},\ldots,A_{m,j})$ be the matrix product appearing in the
  left-hand side of~\eqref{eqn:goalhere}.  Let $t = |T|$. By first
  expanding on the definition of $F$ and then distributing the product
  over the sum we get
  \begin{equation}
  A_T = \prod_{j \in T} \sum_{S \subseteq [m]} \hat{f}(S) \bigotimes_{i
    \in [m]} (A_{i,j})^{S(i)}
  =
  \sum_{S : T \rightarrow 2^{[m]}} \prod_{j \in T} \hat{f}(S(t))
  \prod_{j \in T} \bigotimes_{i \in [m]} (A_{i,j})^{S(t)(i)}.
  \label{eqn:sllsl}
  \end{equation}
  For fixed $T \subseteq [r]$ and $S : T \rightarrow 2^{[m]}$, let
  $B_{T,S} = \prod_j \bigotimes_i (A_{i,j})^{S(t)(i)}$ be the
  matrix product appearing in the right-hand side
  of~\eqref{eqn:sllsl}.  By distributing $\prod$ over $\bigotimes$ and
  applying $A_{i,j} = U_i^* D_{i,j} U_i$ in~\eqref{eqn:sllsl} we get
  \begin{align}
    B_{T,S} & = \bigotimes_{i \in [m]} \Biggl({\prod_{j \in T} (U_i^*
      D_{i,j} U_i)}\Biggr)^{S(t)(i)} = \bigotimes_{i \in [m]}
    \Biggl({U_i^* \Biggl({\prod_{j \in T}
        D_{i,j}}\Biggr)^{S(t)(i)} U_i}\Biggr).
  \label{eqn:sdfss}
  \end{align}
%   For $U = \bigotimes_{i \in [m]} U_i$, the right-hand side
%   of~\eqref{eqn:sdfss} equals
%   \begin{equation}
%    U^* \Biggl( { \bigotimes_{i \in [m]} \Biggl({\prod_{j \in T}
%     D_{i,j}}\Biggr)^{\delta_{S(t)}(i)} }\Biggr) U,
%   \end{equation}
%  \begin{align}
%  \sum_{S : T \rightarrow 2^{[m]}} \prod_{j \in T} \hat{f}(S(t))
%  \bigotimes_{i \in [m]} \Biggl({\prod_{j \in T}
%    A_{i,j}}\Biggr)^{\delta_{S(t)}(i)}
%  & = \sum_{S : T \rightarrow
%    2^{[m]}} \prod_{j \in T} \hat{f}(S(t)) \bigotimes_{i \in [m]}
%  \Biggl({\prod_{j \in T} (U_i^* D_{i,j} U_i)}\Biggr)^{\delta_{S(t)}(i)} =
%  \\ & = \sum_{S : T \rightarrow 2^{[m]}} \prod_{j \in T}
%  \hat{f}(S(t)) \bigotimes_{i \in [m]} \Biggl({U_i^* \Biggl({\prod_{j
%        \in T} D_{i,j}}\Biggr)^{\delta_{S(t)}(i)} U_i}\Biggr) = \\ & =
%  U^* \Biggl( { \sum_{S : T \rightarrow 2^{[m]}} \prod_{j \in T}
%    \hat{f}(S(t)) \bigotimes_{i \in [m]} \Biggl({\prod_{j \in T}
%    D_{i,j}}\Biggr)^{\delta_{S(t)}(i)} }\Biggr) U, \label{eqn:last} \\
% \end{align}
  Hence
  \begin{equation}
  A_T = U^* \Biggl({\sum_{S : T \rightarrow 2^{[m]}} \prod_{j \in T}
    \hat{f}(S(j)) \bigotimes_{i \in [m]} \Biggl({ \prod_{j \in T}
      D_{i,j} }\Biggr)^{S(t)(i)}}\Biggr) U, \label{eqn:last}
  \end{equation}
  for $U = \bigotimes_{i \in [m]} U_i$.  Let $M$ denote the matrix
  sitting within $U^*$ and $U$ in line~\eqref{eqn:last}.  As each
  $D_{i,j}$ is a $d_i \times d_i$ diagonal matrix, $M$ is a $d \times
  d$ diagonal matrix with $d = \prod_{i \in [m]} d_i$.  We think of
  the entries in the diagonal of $M$ as indexed by tuples $c =
  (c_1,\ldots,c_m)$ from $[d_1] \times \cdots \times [d_m]$.  Let
  $M(c)$ denote the entry in position $c$ of the diagonal of $M$. Then
\begin{equation}
  M(c) = \sum_{S : T \rightarrow 2^{[m]}} \prod_{j \in T}
  \hat{f}(S(t)) \prod_{i \in [m]} \prod_{j \in T} \left({D_{i,j}(c_i)}\right)^{S(t)(i)}. \label{eqn:wc}
\end{equation}
Factoring back the product over $j \in T$, the right-hand side
in~\eqref{eqn:wc} reads
\begin{equation}
\prod_{j \in T} \sum_{S \subseteq [m]} \hat{f}(S)
\prod_{i \in [m]} (D_{i,j}(c_i))^{S(i)}
=
\prod_{j \in T} f(D_{1,j}(c_1),\ldots,D_{m,j}(c_m)). \label{eqn:wc2}
\end{equation}
For fixed $j \in [r]$ and $c \in [d_1] \times \cdots \times
[d_m]$, let $X_{j,c} = f(D_{1,j}(c_1),\ldots,D_{m,j}(c_m))$
so that equations~\eqref{eqn:wc} and~\eqref{eqn:wc2} give
$M(c) = \prod_{j \in T} X_{j,c}$.
From~\eqref{eqn:tuplex} and the fact that $f$ is a Boolean
closure operator of $R$, the tuple $(X_{1,c},\ldots,X_{r,c})$
belongs to the relation $R$. Thus
\begin{equation}
\sum_{T \subseteq [r]} \hat{R}(T) M(c) =
\sum_{T \subseteq [r]} \hat{R}(T) \prod_{j \in T} X_{j,c}
= P_R(X_{1,c},\ldots,X_{r,c}) = -1.
\end{equation}
Since this holds for every diagonal entry of $M$,
we get $\sum_{T \subseteq [r]} \hat{R}(T) M = -I$.
Putting it all together, the left-hand side of our
goal~\eqref{eqn:goalhere} evaluates to
\begin{equation}
\sum_{T \subseteq [r]} \hat{R}(T) U^* M U
= U^* \Biggl({ \sum_{T \subseteq [r]} \hat{R}(T) M }\Biggr) U
= U^* (-I) U = -I.
\end{equation}
This gives~\eqref{eqn:goalhere} as desired.

In order to prove the `moreover' clause of the theorem, observe that if
$S \subseteq [m]$ and $(a_1,\ldots,a_m) \in \{ \pm 1 \}^m$, then
\begin{equation}
\bigotimes_{i \in [m]} (a_iI)^{S(i)} = \Biggl({\prod_{i \in S}
  a_i}\Biggr) I,
\end{equation}
where in the left hand side the identity matrices have dimensions
$d_1,\ldots,d_m$, respectively, and in the right-hand side the
identity matrix has dimension $d \times d$ for $d = \prod_{i \in [m]}
d_i$.  It follows that
\begin{equation}
F(a_1 I,\ldots,a_m I) = \Biggl({\sum_{S \subseteq [m]} \hat{f}(S)
\prod_{i \in S} a_i}\Biggr) I = f(a_1,\ldots,a_m) I.
\end{equation}
This completes the proof of the theorem.
\end{proof}

\subsection{Finale}

% We need the following special case of the classical result of Geiger
% that states that if a relation is invariant under all polymorphisms of
% a base set, then it is pp-definable from the relations in the base
% set. Indeed, the converse is also true.
% 
% \begin{theorem}[\cite{Geiger1968}] \label{thm:geigeretal} 
% Let $A$ be a Boolean constraint language and
% let $R$ be a Boolean relation. The following are equivalent:
% \begin{enumerate} \itemsep=0pt
% \item $R$ is pp-definable without constants from $A$,
% \item $R$ is invariant under all Boolean closure operations of $A$.
% \end{enumerate}
% \end{theorem}

\noindent Before we prove Theorem~\ref{thm:collapse}, we need the
following straightforward fact about the role of constants in
pp-definitions.

\begin{lemma} \label{lem:constants}
  Let $A$ be a Boolean constraint language, let $R$ be Boolean a
  relation, and let $A^+ = A \cup \{\{+1\},\{-1\}\}$. The following
  two statements hold.
\begin{enumerate} \itemsep=0pt
\item $R$ is pp-definable from $A$ if and only if it is pp-definable
without constants from $A^+$.
\item $R$ is pp$^*$-definable from $A$ if and only if it is
  pp$^*$-definable without constants from $A^+$.
\end{enumerate}
\end{lemma}

\begin{proof}
  In both cases, for the `only if' part it suffices to replace each
  occurrence of a constant in the quantifier-free part of the
  pp-formula by a new existentially quantified variable $Z$, and force
  it to belong to the corresponding new unary relation in $A^+$ by an
  additional conjunct: if $Z$ replaces the constant $-1$, we force $Z$
  to belong $\{-1\}$ by a new conjunct, and it $Z$ replaces the
  constant $+1$, we force it to belong $\{+1\}$ by a new conjunct. In
  both cases too, the `if' part follows from the reverse construction:
  replace each occurrence of a variable that appears within the scope
  of one of the new unary relations in $A^+$ by the corresponding
  constant, and remove the conjuncts that involve the new unary
  relations. That these transformations are correct follows directly
  from the definitions and the fact that both $I$ and $-I$ commute
  with any operator.
\end{proof}

We are ready to prove Theorem~\ref{thm:collapse}. 

\begin{proof}[Proof of Theorem~\ref{thm:collapse}]
  Assume $R$ is pp$^*$-definable from $A$. By
  Lemma~\ref{lem:constants}, the relation $R$ is also pp$^*$-definable
  without constants from $A^+ = A \cup \{\{+1\},\{-1\}\}$.  Let $r$ be
  the arity of~$R$ and let $\phi(x_1,\ldots,x_r)$ be the pp-formula
  without constants that pp$^*$-defines $R$ from~$A^+$. By Geiger's
  Theorem~\ref{thm:geigeretal} and Lemma~\ref{lem:constants} it
  suffices to show that $R$ is invariant under all Boolean closure
  operations of $A^+$.  

  Let $f : \{ \pm 1 \}^m \rightarrow \{ \pm 1 \}$ be a Boolean closure
  operation of $A^+$. By Theorem~\ref{thm:closure}, the function $F$
  is a closure operation of ${A^+}^*$. Let
  $(a_{1,1},\ldots,a_{1,r}),\ldots,(a_{m,1},\ldots,a_{m,r})$ be tuples
  in $R$ and let $a_j = f(a_{1,j},\ldots,a_{m,j})$ for every $j \in
  [m]$. We need to show that $(a_1,\ldots,a_r)$ is also in $R$.  Let
  $\psi(x_1,\ldots,x_r,y_1,\ldots,y_s)$ be the quantifier-free part of
  $\phi$ and consider the instance over $A^+$ that is given by
\begin{equation}
\psi(x_1/a_{i,1},\ldots,x_r/a_{i,r},y_1/Y_1,\ldots,y_s/Y_s) \label{eqn:inshere}
\end{equation}
as described in the begining of this section.  Since the tuple
$(a_{i,1},\ldots,a_{i,r})$ is in $R$ and $\phi$ pp$^*$-defines $R$,
the instance in~\eqref{eqn:inshere} is satisfiable via operator
assignments over a finite-dimensional Hilbert space for every $i \in
[m]$. Let $B_{i,1},\ldots,B_{i,s}$ be such a satisfying operator
assignment for every $i \in [m]$. Since $I$ and $-I$ commute with any
operator, this means that
$a_{i,1}I,\ldots,a_{i,r}I,B_{i,1},\ldots,B_{i,s}$ is a satisfying
operator assignment of
\begin{equation}
\psi(x_1/X_1,\ldots,x_r/X_r,y_1/Y_1,\ldots,y_s/Y_s) \label{eqn:insthere}
\end{equation}
for every $i \in [m]$. Let $A_j = F(a_{1,j} I,\ldots,a_{m,j} I)$ and
$B_j = F(B_{1,j},\ldots,B_{m,j})$. As $F$ is a closure operation of
${A^+}^*$, the tuple $A_1,\ldots,A_r,B_1,\ldots,B_s$ is a satisfying
operator assignment for~\eqref{eqn:insthere}. Moreover, from the
`moreover' clause in Theorem~\ref{thm:closure} we know that $A_j =
f(a_{1,j},\ldots,a_{m,j}) I = a_j I$ for every $j \in [m]$. Thus, the
instance
\begin{equation}
\psi(x_1/a_1,\ldots,x_r/a_r,y_1/Y_1,\ldots,y_s/Y_s)
\end{equation}
is satisfiable via operator assignments over a finite-dimensional
Hilbert space; the finite-dimensional operator assignment
$B_1,\ldots,B_s$ satisfies it.  As $\phi$ pp$^*$-defines $R$, it
follows that $(a_1,\ldots,a_r)$ is in $R$, as was to be shown.
\end{proof}

%% file: acknowledgments-long.tex
\noindent\textbf{Acknowledgments.} We are grateful to Heribert Vollmer
for sharing with us Steffen Reith's diagram of Post's lattice
(Figure~\ref{fig:postslattice}). This work was initiated and part of the research was carried out while all three 
authors were in residence at the Simons Institute for the Theory of
Computing during the fall of 2016, where they  participated in the program on  Logical Structures in Computation. The research of Albert Atserias was partially funded by the European Research Council (ERC)
under the European Union's Horizon 2020 research and innovation
programme, grant agreement ERC-2014-CoG 648276 (AUTAR), and by MINECO
through TIN2013-48031-C4-1-P (TASSAT2); the research of Simone Severini was partially funded by The Royal Society,
Engineering and Physical Sciences Research Council (EPSRC), and the National
Natural Science Foundation of China (NSFC).